\documentclass[lettersize,journal]{IEEEtran}
\usepackage{amsmath}
\usepackage{amssymb}
\usepackage{amsfonts}
\usepackage{amsthm}
\usepackage{array}
\usepackage{textcomp}
\usepackage{stfloats}
\usepackage{url}
\usepackage{verbatim}
\usepackage{graphicx}
\usepackage{cite}
\usepackage{xspace}
\usepackage{booktabs}
\usepackage{tikz}
\usepackage{pgfplots}
\usetikzlibrary{patterns, patterns.meta}
\usepackage[x11names,table]{xcolor}
\usepackage{pgfplotstable}
\usepackage{subfig}
\usepackage[shortlabels, inline]{enumitem}
\usepackage{multirow}
\usepackage[ruled, vlined, linesnumbered]{algorithm2e}

\usepackage{etoolbox} 

\let\oldcolor\color
\let\oldtextcolor\textcolor

\AtBeginDocument{
  \renewcommand{\textcolor}[2]{#2}
  \renewcommand{\color}[1]{\ignorespaces}
}

\AtBeginEnvironment{tikzpicture}{
  \let\textcolor\oldtextcolor
  \let\color\oldcolor
}
\AtBeginEnvironment{pgfpicture}{ 
  \let\textcolor\oldtextcolor
  \let\color\oldcolor
}

\newenvironment{customlegend}[1][]{%
    \begingroup
    \csname pgfplots@init@cleared@structures\endcsname
    \pgfplotsset{#1}%
}{%
    \csname pgfplots@createlegend\endcsname
    \endgroup
}%

\def\addlegendimage{\csname pgfplots@addlegendimage\endcsname}

\newcommand{\stitle}[1]{\vspace{1ex} \noindent{\bf #1}}
\newtheorem{theorem}{Theorem}
\newtheorem{definition}{Definition}

\newtheorem{lemma}{Lemma}

\newtheorem{example}{Example}

\let\oldnl\nl
\newcommand{\nonl}{\renewcommand{\nl}{\let\nl\oldnl}}

\newcommand{\algo}{APP\xspace}
\newcommand{\algot}{ASRP\xspace}

\definecolor{MC-color}{HTML}{ff7f00}
\definecolor{PI-color}{HTML}{ec028d}
\definecolor{FP-color}{HTML}{007f7f}
\definecolor{first-color}{HTML}{00adee}
\definecolor{second-color}{HTML}{ff0000}

\definecolor{APPS-color}{HTML}{aa4e32}
\definecolor{APPE-color}{HTML}{b44df2}

\definecolor{AHPP-color}{HTML}{00b9f2}
\definecolor{BHPP-color}{HTML}{f8c6a9}
\definecolor{HPP-color}{HTML}{b4a6ca}
\definecolor{PPR-color}{HTML}{924361}
\definecolor{Simrank-color}{HTML}{66ce63}
\definecolor{Pearson-color}{HTML}{ae5e52}
\definecolor{Jaccard-color}{HTML}{ff7f0e}
\definecolor{AnchorGNN-color}{HTML}{e7991b}
\definecolor{GEBEp-color}{HTML}{ea6a4f}
\definecolor{BiANE-color}{HTML}{3a8b9e}
\definecolor{RW_Uniform-color}{HTML}{d93189}
\definecolor{IDBR-color}{HTML}{ffee65}
\definecolor{BiNE-color}{HTML}{d4eded}

\definecolor{AHPP-line-color}{HTML}{ff0000}
\definecolor{BHPP-line-color}{HTML}{ff7f00}
\definecolor{HPP-line-color}{HTML}{ec028d}
\definecolor{PPR-line-color}{HTML}{007f7f}
\definecolor{Simrank-line-color}{HTML}{00adee}
\definecolor{Pearson-line-color}{HTML}{aa4e32}
\definecolor{Jaccard-line-color}{HTML}{b44df2}
\definecolor{AnchorGNN-line-color}{HTML}{16048a}
\definecolor{GEBEp-line-color}{HTML}{0162a7}
\definecolor{BiANE-line-color}{HTML}{45c4b0}
\definecolor{RW_Uniform-line-color}{HTML}{fb9575}
\definecolor{IDBR-line-color}{HTML}{0000ff}
\definecolor{BiNE-line-color}{HTML}{bfbf00}

\begin{document}

\title{Scalable Similarity Search over Large Attributed Bipartite Graphs} 
	\author{Xi Ou, Longlong Lin, Zeli Wang, Pingpeng Yuan, Rong-Hua Li

	\IEEEcompsocitemizethanks{

   \IEEEcompsocthanksitem Xi Ou and Longlong Lin are with the College of Computer and Information Science, Southwest University. Email: xiou.cs@outlook.com, longlonglin@swu.edu.cn.  Corresponding Author:  Longlong Lin
        \IEEEcompsocthanksitem  Zeli Wang is with Chongqing University of Posts and Telecommunications.
		Email: zlwang@cqupt.edu.cn.
        \IEEEcompsocthanksitem  Pingpeng Yuan is with the School of Computer Science and Technology, Huazhong University of Science and Technology, China.
		E-mail: ppyuan@hust.edu.cn.
        				\IEEEcompsocthanksitem  Rong-Hua Li is with 
		Beijing Institute of Technology.  Email: lironghuabit@126.com.}
    }

\maketitle

\begin{abstract}
Bipartite graphs are widely used to model relationships between entities of different types, where nodes are divided into two disjoint sets. Similarity search, a fundamental operation that retrieves nodes similar to a given query node, plays a crucial role in various real-world applications, including machine learning and graph clustering. However, existing state-of-the-art methods often struggle to accurately capture the unique structural properties of bipartite graphs or fail to incorporate the informative node attributes, leading to suboptimal performance. Besides, their high computational complexity limits scalability, making them impractical for large graphs with millions of nodes and tens of thousands of attributes.  To overcome these challenges, we first introduce Attribute-augmented Hidden Personalized PageRank (AHPP), a novel random walk model designed to blend seamlessly both the higher-order bipartite structure proximity and attribute similarity. We then formulate the similarity search over attributed bipartite graphs as an approximate AHPP problem and propose two efficient push-style local algorithms with provable approximation guarantees. \textcolor{blue}{Finally, extensive experiments on real-world and synthetic datasets validate the effectiveness of AHPP and the efficiency of our proposed algorithms when compared with fifteen competitors.}
\end{abstract}


\section{Introduction} \label{sec:introduction}
\textcolor{blue}{Similarity search on graphs is a fundamental technique that aims to identify nodes similar to a given query node, which has witnessed numerous applications, including graph clustering \cite{DBLP:conf/kdd/Lin0WZL24, DBLP:conf/aaai/LinLJ23, DBLP:journals/tnn/LiRGYY25, DBLP:journals/tmm/LiGYYR25}, community search \cite{li2023persistent, li2024maximal, DBLP:conf/IEEEwisa/LiXD24, DBLP:journals/pvldb/LinYLZQJJ24, 11235568}, and graph neural networks \cite{DBLP:journals/pvldb/MengLLLW24, DBLP:conf/mir/YuLLWOJ24,JAS,DBLP:conf/mir/LiuLYOZY025}.}  Despite the significant success, most existing approaches are designed specifically for unipartite graphs, where all nodes belong to a single type \cite{DBLP:journals/corr/abs-2412-10789,wang2017fora, wang2021exactsim, liao2022efficient, liao2023efficient}. However, in many real-world applications, nodes naturally fall into two distinct types, with interactions occurring only between different types. For instance, in an actor-movie network, nodes represent actors or movies, and edges indicate an actor's participation in a movie. In a customer-product network, nodes represent customers and products, with edges signifying purchases. These networks are commonly called bipartite graphs \cite{huang2020biane, yang2022scalable, wu2023billion, yang2024effective}. Consequently, unipartite graph similarity search methods ignore the unique properties of bipartite graphs with two disjoint vertex sets, leading to suboptimal results \cite{yang2022efficient, liu2024bird}.



Recently, several studies have explored similarity search specifically designed for bipartite graphs (Section \ref{sec:relate}). A notable example is the Hidden Personalized PageRank (HPP), a variant of Personalized PageRank (PPR) \cite{wang2017fora, liao2022efficient, liao2023efficient} designed for bipartite graphs. HPP and its variants have demonstrated exceptional performance across various bipartite graph applications \cite{yang2022efficient, liu2024bird}. \textcolor{blue}{However, most existing similarity models primarily focus on the structural properties of bipartite networks, while largely overlooking informative node attributes. As noted in prior studies \cite{shi2019robust, ahmad2020conceptual}, real-world graphs are often constructed from noisy and incomplete data, which may contain spurious edges or suffer from missing links, thereby limiting the effectiveness of structure-only methods. In contrast, nodes in practical applications are usually associated with rich attribute information that can serve as an important complement to the imperfect topological structure. For instance, in an e-commerce recommendation scenario modeled as a user-item bipartite graph, explicit user-item interactions (e.g., purchases or ratings) are typically sparse and incomplete, especially for new or inactive users. Nevertheless, users’ demographic profiles and behavioral attributes, as well as items’ textual descriptions and categorical features, provide valuable semantic cues for identifying similar users or items even in the absence of sufficient structural connections. Similar phenomena also arise in citation networks, where papers with few citations may still be accurately related through textual content or topical attributes. These observations indicate that effectively incorporating node attributes is crucial for robust similarity measurement in real-world bipartite graphs. Thus, developing models that can perform similarity search on attributed bipartite graphs is both necessary and timely.}




Given an attributed bipartite graph \( G \) with two disjoint node partitions, \( U \) and \( V \), and a query node \( u \in U \), similarity search over \( G \) aims to identify nodes in \( U \) that are similar to \( u \) based on a well-defined similarity measure \cite{sun2005relevance} (Definition \ref{def:ss}).  One simple solution is to seek the idea of node embedding. That is, we can first encode all nodes into low-dimensional embedding vectors and then calculate the distance of the embedding vectors of the two nodes as their similarity. \textcolor{blue}{However, such solutions (e.g., \cite{yang2022scalable, huang2020biane, wu2023billion, DBLP:conf/sigir/LiWLLJD24, DBLP:journals/tkde/GaoHCLZZ22}) have prohibitive training costs, resulting in poor scalability. Meanwhile, they calculate the similarity through the intermediate embedding vectors, causing the high-order information to be ignored and greatly compressing qualities (Section \ref{sec:experiments}).} In essence, two crucial observations underlie an effective similarity search in attributed bipartite graphs. Firstly, an effective similarity measure should capture both direct and indirect interactions (i.e., higher-order relationships) among nodes, by exploiting the cross-partition connections defined by the bipartite structure. Specifically, nodes in \( U \) that share multi-hop connections via nodes in \( V \) should be considered more similar \cite{DBLP:conf/iclr/TadicBN25}. For instance, in an e-commerce scenario, two users in \( U \) are not merely deemed more similar due to sharing common interacted products (i.e., 1-hop neighbors in \( V \)), but rather when their shopping behaviors exhibit intersecting multi-hop connection patterns through the product space \( V \). Secondly, the similarity measure should integrate attribute-based associations between nodes. More precisely, each node \( u\in U \) is associated with a set of attributes that offer complementary information for identifying similar nodes. The underlying intuition is that similar nodes in \( U \) tend to form multi-hop connections through attributes \cite{li2023attributed}. For example, in a retail scenario, two young female users are deemed more similar if their multi-hop paths via attributes such as "age-youth" and "gender-female" intersect across multiple hops in the bipartite graph. As a consequence, incorporating both structural interactions and attribute-based associations is essential for improving similarity search in attributed bipartite graphs.

Based on these observations, we introduce a novel similarity measure called \textit{Attribute-augmented Hidden Personalized PageRank (AHPP)}. AHPP, as detailed in Section \ref{subsec:ahpp}, extends Hidden Personalized PageRank (HPP) by seamlessly integrating both higher-order attribute similarity and structural similarity into a unified framework. Specifically, HPP alone is insufficient for similarity search in attributed bipartite graphs, as it fails to leverage attribute information effectively. To address this limitation, AHPP enhances similarity estimation by combining the newly proposed attribute transition matrix (Eq. \eqref{eq:compute_attributed_transition}) with the hidden transition matrix (Eq. \eqref{eq:compute_hidden_transition}) introduced in HPP. Additionally, AHPP can be interpreted as the convergence probability of the \( \alpha \)-attribute-augmented hidden random walk (Lemma \ref{lem:compute_A-RWR}), which further enhances its interpretability for practical applications. Existing methods for solving AHPP (Section \ref{sec:existing}) typically require materializing an intermediate weighted graph with \( O(|U|^2) \) time complexity, or rely on expensive Monte Carlo sampling or local graph diffusion. These approaches become impractical for large-scale graphs with millions of nodes and tens of thousands of attributes. Consequently, existing algorithms cannot be directly applied to efficiently solve the proposed AHPP problem, posing a significant efficiency challenge. In short, the main contributions are summarized as follows:


\stitle{Novel Problem.} We introduce a novel model AHPP and frame the similarity search over attributed bipartite graphs as the \( \epsilon \)-approximate single-source AHPP query (Section \ref{subsec:ahpp}). Besides, none of the existing PPR-like techniques can be directly applied to efficiently solve our problem (Section \ref{sec:existing}).

\stitle{Efficient Algorithmic Design.} We propose two efficient algorithms, namely APP (Alternating Propagation Push, Section \ref{subsec:APP}) and ASRP (Adaptive Synchronous Residue Push, Section \ref{subsec:ASRP}), along with their theoretical analysis, to efficiently solve the \( \epsilon \)-approximate single-source AHPP query problem. In particular, the APP algorithm builds on two key observations from the well-known Forward Push technique \cite{wang2017fora}: (i) the equivalence of residue propagation and (ii) redundancy in residue propagation. The ASRP algorithm extends APP by incorporating a synchronous push strategy and a more effective termination threshold, achieving a near-linear time complexity.

\stitle{Extensive Empirical Validation.} We conduct extensive experiments on real-world and synthetic graphs to validate the effectiveness and efficiency of our proposed solutions against thirteen competitors. The results demonstrate that the AHPP model consistently outperforms existing similarity search models in cluster consistency validation, top-k precision, and link prediction. \textcolor{blue}{Notably, our AHPP achieves F1-score improvements ranging from 8\% to 12\% compared with state-of-the-art methods in the clustering consistency validation task.} Besides, the proposed AHPP and ASRP are 1$\sim$2 orders of magnitude faster than existing methods.

\stitle{Reproducibility.}  The source code, datasets, and parameter settings are available at https://github.com/longlonglin/AHPP.


%
\vspace{-0.2cm}
\section{PRELIMINARIES} \label{sec:pro}

\subsection{Basic Notations}
Let $G=(U, V,  E, \mathcal{A}, E_\mathcal{A})$ be an attributed bipartite graph, where $U$ and $V$ are two disjoint node sets (i.e., $U\cap V =\emptyset$). Each edge $e = (u, v, w(u, v)) \in E$ represents a connection between a node in $U$ and a node in $V$, associated with a positive weight $w(u, v)$. For each node $u$, we use $N(u)=\{v|(u,v)\in E\}$ to denote the neighbors of $u$ and $d(u)=\sum_{v \in N(u)} w(u, v)$ to represent the (weighted) degree of $u$. $\mathcal{A}$ is the set of attributes associated with vertices for describing node profiles. For any node $u$ and attribute $a \in \mathcal{A}$, we let $\mathcal{A}(u) \subseteq \mathcal{A}$ be the set of attributes associated with $u$ and $\mathcal{A}^{-1}(a) = \{u|a \in \mathcal{A}(u)\}$ be the set of nodes that possess attribute $a$. $E_\mathcal{A}$ stands for the set of node-attribute associations, where each element is a tuple $(u, a, w(u, a))$ that signifies node $u$ is associated with attribute $a \in \mathcal{A}$ with a weight $w(u, a)$ (i.e., the attribute value). Note that, by default, we consider $U$ as the target node set and evaluate the similarity of the nodes within it. A high-level definition of the similarity search over the attributed bipartite graph problem is stated as follows.
\begin{definition} [Similarity Search over Attributed Bipartite Graphs \cite{sun2005relevance}] \label{def:ss}
	Given an attributed bipartite graph $G$ and the target node set  $U$, the similarity search of any node $u \in U$ aims to identify nodes within $U$ that are similar to $u$ in terms of both structure proximity and attribute similarity.
\end{definition}
In this paper, we denote matrices and vectors using bold uppercase and lowercase letters, e.g., $\mathbf{M}$ and $\mathbf{x}$ respectively. Accordingly, $\mathbf{M}[i]$ (resp. $\mathbf{M}[:, j]$) signifies the $i$-th row (resp. the $j$-th column) vector of $\mathbf{M}$, and $\mathbf{M}[i, j]$ signifies the element at the $i$-th row and the $j$-th column of $\mathbf{M}$. We denote the hidden transition matrix \cite{liu2024bird} for the node set $U$ as $\mathbf{P_S}$, where each entry $(u_i, u_j)$ is calculated by
\begin{equation}\label{eq:compute_hidden_transition}
    \mathbf{P_S}[u_i, u_j] = \sum_{v \in N(u_i) \cap N(u_j)}{\frac{w(u_i, v)}{d(u_i)}\cdot \frac{w(v, u_j)}{d(v)}}. 
\end{equation}
Note that the number of non-zero entries in $\mathbf{P_S}$ can reach up to $O(|U|^2)$ in the worst case where one node $v \in V$ has $O(|U|)$ neighbors in $U$, resulting in poor scalability.

\begin{figure}[t!]
	\centering
	\includegraphics [width=0.45\textwidth] {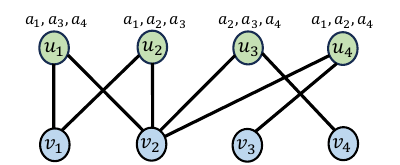}
	\caption{The example of an attributed bipartite graph.}
\label{fig:preliminary_example} \vspace*{-0.3cm}
\end{figure}

\begin{example}
Figure \ref{fig:preliminary_example} illustrates an example of an attributed bipartite graph \( G \). The top part of the figure (in light green) represents the set \( U \), while the bottom part (in light blue) represents the set \( V \). Specifically, the graph contains four nodes \( u_1 \) to \( u_4 \) with attributes \( a_1 \) to \( a_4 \) in \( U \), and four nodes \( v_1 \) to \( v_4 \) in \( V \). For instance, node \( u_1 \) has two neighbors, \( v_1 \) and \( v_2 \), and is associated with three attributes: \( a_1 \), \( a_3 \), and \( a_4 \). Assume that all edge weights in \( G \) are set to 1. According to Eq. \eqref{eq:compute_hidden_transition}, the transition probability \( \mathbf{P_S}[u_1, u_2] \) is calculated as  $\mathbf{P_S}[u_1, u_2] = \frac{w(u_1, v_1)}{d(u_1)} \cdot \frac{w(v_1, u_2)}{d(v_1)} + \frac{w(u_1, v_2)}{d(u_1)} \cdot \frac{w(v_2, u_2)}{d(v_2)} = \frac{1}{2} \cdot \frac{1}{2} + \frac{1}{2} \cdot \frac{1}{4} = \frac{3}{8}$.
Analogously, $\mathbf{P_S}[u_1, u_3] = \frac{w(u_1, v_2)}{d(u_1)} \cdot \frac{w(v_2, u_3)}{d(v_2)} = \frac{1}{2} \cdot \frac{1}{4} = \frac{1}{8}$. It can be observed that node \( u_1 \) exhibits a higher transition probability to node \( u_2 \) than to node \( u_3 \). This observation aligns with the intuition that \( u_1 \) is more strongly connected to \( u_2 \) via multiple intermediate nodes in \( V \), specifically \( v_1 \) and \( v_2 \).
\end{example}

\subsection{Hidden Personalized PageRank} \label{subsec:hpp}
Similarity search is a key task in graph analysis that measures the relevance between two nodes \cite{yang2022efficient}. One of the most prominent models for this purpose is the \textit{Personalized PageRank} (PPR) \cite{liao2022efficient, liao2023efficient}, which is well-regarded for its effectiveness and solid theoretical foundation. As a variant of PPR, the \textit{Hidden Personalized PageRank} (HPP) \cite{yang2022efficient, liu2024bird} is designed for bipartite graphs. Specifically,  given two nodes $u_i, u_j \in U$ and a restart probability $\alpha \in (0, 1)$, the HPP value of $u_j$ with respect to (w.r.t) $u_i$ is defined as the probability that an $\alpha$-hidden random walk starting from $u_i$ stops at $u_j$. The  $\alpha$-hidden random walk of $u_i$ is denoted as follows: (1) it starts from the node $u_i$; (2) at each step it stops at the current node with probability $\alpha$, or it continues to walk according to Eq. \eqref{eq:compute_hidden_transition} with probability $(1 - \alpha)$. Therefore, HPP can be understood as the PPR on the weighted graph $\bar{G}$, which is constructed from $G$, where the node set of $\bar{G}$ is $U$ and the weights of the edges are defined as $\mathbf{P_S}[u_i, u_j]$ for all $u_i, u_j \in U$. However, HPP only considers graph structural information and ignores the informative node attributes, leading to sub-optimal results, as stated in our empirical results (Section \ref{sec:experiments}).

\subsection{From HPP to Attribute-augmented HPP}\label{subsec:ahpp}
In this subsection, we propose a novel similarity measure, \textit{Attribute-augmented  Hidden Personalized PageRank} (AHPP), to model the similarity between two nodes over an attributed bipartite graph. However, designing an effective state transition matrix for an attributed bipartite graph remains an important challenge. Intuitively, nodes in \( U \) that share multi-hop connections via nodes in \( V \) should exhibit high similarity \cite{DBLP:conf/iclr/TadicBN25}. Besides, each node $u \in U$ is associated with a set of attributes that can provide supplementary information to help identify nodes in $U$ that are similar to $u$ \cite{li2023attributed}. Based on these intuitions, we redefine the hidden transition matrix $\mathbf{P_S}$ (i.e., Eq. \eqref{eq:compute_hidden_transition}) as the \textit{structure transition matrix}, which captures the structural information. Following \cite{}, we say that $u_i$ and $u_j$ are connected via attribute $r_k$ if $u_i$ and $u_j$ are associated with a common attribute $r_k$. To further capture the informative node attributes, we define the \textit{attribute transition matrix} $\mathbf{P_A} \in \mathbb{R}^{|U| \times |U|}$, where each element $(u_i, u_j)$ is computed as follows:
\begin{equation}\label{eq:compute_attributed_transition}\scriptsize
\mathbf{P_A}[u_i, u_j] = \sum_{a \in \mathcal{A}(u_i) \cap \mathcal{A}(u_j)} \frac{w(u_i, a)}{\sum_{a_i \in \mathcal{A}(u_i)} w(u_i, a_i)}\cdot \frac{w(u_j, a)}{\sum_{u_k \in \mathcal{A}^{-1}(a)} w(u_k, a)}.
\end{equation}
Intuitively, $\mathbf{P_A}[u_i, u_j]$ represents the transition probability from state $u_i$ to state $u_j$ through any common attribute in $\mathcal{A}(u_i) \cap \mathcal{A}(u_j)$. Specifically, $\frac{w(u_i, a)}{\sum_{a_i \in \mathcal{A}(u_i)} w(u_i, a_i)}$ represents the probability that node $u_i$ chooses attribute $a$, while $\frac{w(u_j, a)}{\sum_{u_k \in \mathcal{A}^{-1}(a)} w(u_k, a)}$ denotes the probability that node $u_j$ is chosen by attribute $a$. We formulate AHPP as follows.


\begin{definition}\label{def:A-RWR}
[Attribute-augmented  Hidden Personalized PageRank (AHPP)]
  \textcolor{blue}{  Given an attributed bipartite graph $G=(U, V,  E, \mathcal{A}, E_\mathcal{A})$, a restart probability $\alpha$, an attribute jumping probability $\beta$, a source node $u_i \in U$ and a target node $u_j \in U$,
the AHPP score of node $u_j$ w.r.t. $u_i$ is denoted by
\begin{equation}\label{eq:random_walk_style}
    \pi(u_i, u_j)=\sum_{\ell=0}^{\infty} \alpha (1 - \alpha)^{\ell} \cdot \mathbf{P}^{\ell}[u_i, u_j],
\end{equation}
    in which $\mathbf{P} = (1 - \beta) \cdot \mathbf{P_S} + \beta \cdot \mathbf{P_A}$.}
\end{definition}

\stitle{Remark.} The novel state transition matrix \( \mathbf{P} \) is a combination of the structure transition matrix \( \mathbf{P_S} \) and the attribute transition matrix \( \mathbf{P_A} \), weighted by the attribute jumping probability \( \beta \). This formulation captures a trade-off between the topological structure and node attributes, enabling us to measure node similarity with distinct characteristics by adjusting \( \beta \) (more details in Section \ref{sec:experiments}). Specifically, when \( \beta < 0.5 \), the similarity measure emphasizes the structural interaction among vertices. In contrast, when \( \beta > 0.5 \), the focus shifts toward attribute similarity among vertices. From Eq. \eqref{eq:random_walk_style}, we see that \( \pi(u_i, u_j) \) succinctly summarizes an infinite number of random walks following \( \mathbf{P} \), quantifying the multi-hop higher-order structural and attribute connections between \( u_i \) and \( u_j \). In other words, \( \pi(u_i, u_j) \) represents the cumulative direct and indirect influence of \( u_i \) on \( u_j \) in terms of both topological structure and node attributes. Furthermore, the following lemma demonstrates that \( \pi(u_i, u_j) \) has a probabilistic interpretation, thereby enhancing its strong interpretability for practical applications. 

\begin{lemma}\label{lem:compute_A-RWR}
    Given an attributed bipartite graph $G=(U, V,  E, \mathcal{A}, E_\mathcal{A})$, a restart probability $\alpha$, an attribute jumping probability $\beta$, a source node $u_i \in U$ and a target node $u_j \in U$.  
    The AHPP score $\pi(u_i, u_j)$ is the probability that the following $\alpha$-attribute-augmented hidden random walk starting from $u_i$ stops at $u_j$.  The $\alpha$-attribute-augmented hidden random walk of $u_i$ is denoted as follows: (1) it starts from the node $u_i$; (2) at each step, the walker either stops at the current node $u_k \in U$ with probability $\alpha$, or with  probability $(1 - \alpha)$, jumps to another node $u_l \in U$ according to the following rules:
    
\stitle{Rule 1: structure transition.} Moving to node $u_l$ based on the structure transition probability $\mathbf{P_S}[u_k, u_l]$ with probability $(1 - \beta)$.

\stitle{Rule 2: attribute transition.} Moving to node $u_l$ based on the attribute transition probability $\mathbf{P_A}[u_k, u_l]$ with probability $\beta$.
\end{lemma}

\begin{proof}
  By the definitions of $\mathbf{P_S}$ and $\mathbf{P_A}$, 
    the probability that node $u_i$ walks  to $u_j$ according to Rules 1 and 2 is  $\mathbf{P}[u_i, u_j] = (1 - \beta) \cdot \mathbf{P_S}[u_i, u_j] + \beta \cdot \mathbf{P_A}[u_i, u_j]$.
    Thus, the probability of $\alpha$-attribute-augmented  hidden random walk starting from $u_i$ and reaching $u_j$ at $\ell$-th step is  $(1 - \alpha)^{\ell} \cdot \mathbf{P}^{\ell}[u_i, u_j]$. Given that the walk terminates at its current node with probability $\alpha$ at each step, the probability of starting from $u_i$ and stopping at $u_j$ at step $\ell$ is $\alpha (1 - \alpha)^{\ell} \cdot \mathbf{P}^{\ell}[u_i, u_j]$. As a result, the probability of the $\alpha$-attribute-augmented hidden random walk starting from $u_i$ stops at $u_j$ is $\sum_{\ell=0}^{\infty}\alpha (1 - \alpha)^{\ell} \cdot \mathbf{P}^{\ell}[u_i, u_j]$, which completes the proof.
\end{proof}

The exact AHPP computation can be reframed as solving a corresponding linear system, which entails a time complexity of $O(n^3)$, where $n$ represents the number of vertices. Consequently, in line with the approaches outlined in \cite{wang2017fora, liao2022efficient, liao2023efficient}, this paper focuses on computing the following approximate single-source AHPP queries. 

\stitle{Our Problem.}
	\textit{($\epsilon$-Approximate Single-Source AHPP Query).}
	\textcolor{blue}{Given an attributed bipartite graph $G=(U, V,  E, \mathcal{A}, E_\mathcal{A})$, a restart probability $\alpha$, an attribute jumping probability $\beta$, a source node $u \in U$, and an error threshold $\epsilon$, \textcolor{blue}{the} $\epsilon$-approximate single-source AHPP query returns an estimated $\hat{\pi}(u, u_i)$ for each node $u_i \in U$, satisfying the condition that
\begin{equation}\label{eq:accuracy_assurance}
     |\pi(u,u_i) - \hat{\pi}(u,u_i)| \leq \epsilon
 \end{equation}
 holds with a constant probability.}

\section{Existing Solutions and Their Defects} \label{sec:existing}

Existing solutions for PPR or HPP calculations cannot be applied directly to the proposed $\epsilon$ approximate single-source AHPP query. This is because they neglect the informative node attribute. Thus, in this section, we adapt several existing methods to suit the $\epsilon$-approximate single-source AHPP query, which can be classified into three distinct groups: Monte Carlo-based, Power Iteration-based, and Forward Push-based. 

\subsection{The Monte Carlo Method} \label{subsec:mc}
The Monte Carlo (MC) method \cite{liao2023efficient} is a well-established approach for deriving sample-based estimates of target values through probabilistic simulations. Inspired by interpreting AHPP as $\alpha$-attribute-augmented hidden random walks (Lemma \ref{lem:compute_A-RWR}), the AHPP vector $\boldsymbol{\pi}_u$ can be effectively approximated using the MC method. Furthermore, MC can effectively provide high probability unbiased and accurate estimates of AHPP. The core principle of the MC method involves simulating many random walks from $u$, utilizing the empirical termination distribution as an approximation for $\boldsymbol{\pi}_u$. This approach offers a natural solution for addressing the $\epsilon$-approximate single-source AHPP query. Assume that $\omega$ random walks are generated independently, requiring an expected time of $O(\omega/\alpha)$. According to \cite{liao2023efficient}, setting $\omega = O\left( \frac{2 \left( 1 + \epsilon / 3 \right) \cdot \ln \left( |U| / p_f \right)}{\epsilon^2} \right)$ ensures that for each $u_i \in U$, we have $|\hat{\pi}(u,u_i) - \pi(u,u_i)| < \epsilon$ holds with at least $1 - p_f$ probability. 

\stitle{Limitations of MC.} While the MC can yield high-probability, unbiased, and accurate estimates of AHPP, it is relatively inefficient due to the need for sampling a large number of random walks. Furthermore, to enable random walk sampling on weighted graphs, MC requires constructing alias structures \cite{liao2023efficient} for each node's neighborhood during the preprocessing phase, resulting in significant computational overhead.

\subsection{The Power Iteration Method} \label{subsec:pi}
The Power Iteration (PI) method \cite{wang2017fora} is a fundamental iterative method for computing the entire AHPP vector $\boldsymbol{\pi}_u$, as originally introduced in Google's seminal work. Specifically, PI estimates AHPP values by iteratively solving the following linear system (a variant of Eq. \eqref{eq:random_walk_style}): $\boldsymbol{\pi}_u = (1-\alpha)\cdot\boldsymbol{\pi}_u\cdot\mathbf{P} + \alpha\cdot\mathbf{e}_u$,
where the one-hot vector $\mathbf{e}_u \in \mathbb{R}^{1 \times |U|}$ takes a value of 1 at the $u$-th position and 0 elsewhere, $\boldsymbol{\pi}_u(u_i) = \pi(u,u_i)$ for all $u_i \in U$, and $\mathbf{P} = (1 - \beta) \cdot \mathbf{P_S} + \beta \cdot \mathbf{P_A}$. Algorithm \ref{alg:PI} provides the pseudo-code for the PI method to approximate $\boldsymbol{\pi}_u$ given the input graph $G$ and the source node $u$. It is noteworthy that, in the worst-case scenario, the number of non-zero entries in the transition matrix $\mathbf{P}$ can scale as $O(|U|^2)$. As a result, the computational cost of matrix-vector multiplications in each iteration is $O(|U|^2)$. According to \cite{wang2017fora}, PI guarantees an absolute error bound of $\epsilon$ for each $\pi(u, u_i)$ after $T = \log_{\frac{1}{1 - \alpha}} \frac{1}{\epsilon}$ iterations of matrix multiplications. Thus, PI resolves the single-source AHPP query within an absolute error bound of $\epsilon$ in $O(|U|^2 \cdot T) = O(|U|^2 \cdot \log(1/\epsilon))$ time.

\stitle{Limitations of PI.} Although PI is straightforward to implement, it becomes inefficient on large graphs due to its \( O(|U|^2) \) time and space complexities. Additionally, when the error threshold (i.e., \( \epsilon \)) is small, PI requires a large number of iterations of matrix-vector multiplication, which can be extremely computationally expensive.

\begin{algorithm}[t!]
    \caption{Power Iteration (PI)}
    \label{alg:PI}
    \KwIn{Attributed bipartite graph $G$, source node $u$, restart probability $\alpha$, attribute jumping probability $\beta$, the number of iterations $T$.}
    \KwOut{$\boldsymbol{\pi}_u$.}
    $\boldsymbol{\pi}_u \gets \mathbf{e}_u$\;
    \For{$i \gets 1$ to $T$} {
        $\boldsymbol{\pi}_u \gets (1 - \alpha) \cdot ((1 - \beta) \cdot \boldsymbol{\pi}_u \cdot \mathbf{P_S} + \beta \cdot \boldsymbol{\pi}_u \cdot \mathbf{P_A}) + \alpha \cdot \mathbf{e}_u$\;
    }
    \Return{$\boldsymbol{\pi}_u$}\;
\end{algorithm}

\begin{algorithm}[t!]
    \caption{Forward Push (FP)}
    \label{alg:FP}
    \KwIn{Attributed bipartite graph $G$, source node $u$, restart probability $\alpha$, attribute jumping probability $\beta$, residue threshold $r_{max}$.}
    \KwOut{$\{ \hat{\pi}(u, u_i) \mid u_i \in U \}$.}
    $r(u, u) \gets 1$; $r(u, u_i) \gets 0 \ \forall{u_i \in U\setminus \{u\}}$\;
    $\hat{\pi}(u, u_i) \gets 0 \ \forall{u_i \in U}$\;
    \While {$\exists u_i \in U$ s.t. $r(u, u_i)> r_{max} \cdot (|N(u_i)| + |\mathcal{A}(u_i)|)$} {
        $\hat{\pi}(u, u_i) \gets \hat{\pi}(u, u_i) + \alpha \cdot r(u, u_i)$\;
        $\rho \gets (1 - \alpha) \cdot r(u, u_i)$\;
        $r(u, u_i) \gets 0$\;
        \For {each $v \in N(u_i)$} {
            \For {each $u_j \in N(v)$} {
                $r(u, u_j) \gets r(u, u_j) + (1 - \beta) \cdot \frac{w(u_i, v) \cdot w(v, u_j)}{d(u_i) \cdot d(v)} \cdot \rho$\;
            }
        }
        \For {each $a \in \mathcal{A}(u_i)$} {
            \For {each $u_j \in \mathcal{A}^{-1}(a)$} {
                $r(u, u_j) \gets r(u, u_j) + \frac{\beta \cdot w(u_i, a) \cdot w(u_j, a)\cdot \rho}{\sum_{a_x \in \mathcal{A}(u_i)}{w(u_i, a_x)} \cdot \sum_{u_x \in \mathcal{A}^{-1}(a)}{w(u_x, a)}}$\;
            }
        }
    }
    \Return{$\{ \hat{\pi}(u, u_i) \mid u_i \in U \}$}\;
\end{algorithm}

\subsection{The Forward Push Method} \label{subsec:fp}
The Forward Push (FP) method  \cite{wang2017fora} is a local-push method capable of answering single-source AHPP queries without searching the whole graph. The fundamental idea behind FP is to simulate numerous random walks deterministically by distributing the probability mass from a node to its neighbors. Specifically, given a source node $u$ and a parameter $r_{max} \in [0, 1]$, FP maintains the following information throughout the execution process:
\begin{itemize}
    \item a reserve $\hat{\pi}(u, u_i)$ for all $u_i \in U$: the probability mass that remains at node $u_i$, it is an underestimate of $\pi(u, u_i)$;
    \item a residue $r(u, u_i)$ for all $u_i \in U$: \textcolor{blue}{The probability mass on node $u_i$ will be redistributed to other nodes. In a random walk, $r(u, u_i)$ represents the mass from a walk starting at $u$ that is still alive at $u_i$, where a walk is alive if it has not yet terminated.}
\end{itemize}

\textcolor{blue}{Algorithm~\ref{alg:FP} presents the FP method. It initializes the residue $r(u, u) = 1$ for the source node $u$ and $0$ for all other nodes, and sets the approximate AHPP values $\hat{\pi}(u, u_i) = 0$ for all $u_i \in U$ (Lines 1--2). Iteratively, residues of selected nodes are distributed to their two-hop neighbors and nodes sharing attributes. For a node $u_i$ with large residue ($r(u, u_i) > r_{max} \cdot (|N(u_i)| + |\mathcal{A}(u_i)|)$), the AHPP value is updated as $\hat{\pi}(u, u_i) \mathrel{+}= \alpha \cdot r(u, u_i)$. Each two-hop neighbor $u_j \in \bigcup_{v \in N(u_i)} N(v)$ receives $(1 - \alpha)(1 - \beta) \frac{w(u_i, v) w(v, u_j)}{d(u_i) d(v)} r(u, u_i)$ (Lines 7--9), and each node $u_k$ sharing attributes with $u_i$ receives $(1 - \alpha)\beta \frac{w(u_i, a) w(u_k, a)}{\sum_{a_x \in \mathcal{A}(u_i)} w(u_i, a_x) \sum_{u_x \in \mathcal{A}^{-1}(a)} w(u_x, a)} r(u, u_i)$ (Lines 10--12). After processing, $r(u, u_i)$ is set to zero (Line 6). These steps repeat until all nodes satisfy $r(u, u_i) \le r_{max} \cdot (|N(u_i)| + |\mathcal{A}(u_i)|)$.}

\begin{lemma}\label{lem:fwd_time_complexity}
Given the query node $u$, the time complexity of Algorithm \ref{alg:FP} is $O(\sum_{u_i \in U}\frac{\pi(u, u_i)\cdot (\sum_{v \in N(u_i)}{|N(v)|} + \sum_{a \in \mathcal{A}(u_i)}{|\mathcal{A}^{-1}(a)|})}{\alpha \cdot r_{max} \cdot(|N(u_i)| + |\mathcal{A}(u_i)|)}).$
\end{lemma}


\begin{proof}
In each iteration for a selected node $u_i$ (Line 3), a portion $\alpha$ of the residue $r(u, u_i) > r_{max} \cdot (|N(u_i)| + |\mathcal{A}(u_i)|)$ is converted to its approximate AHPP $\pi(u, u_i)$, while the remaining residue is distributed among two-hop neighbors $\bigcup_{v \in N(u_i)}{N(v)}$ as well as nodes sharing attributes with $u_i$. Since $\hat{\pi}(u, u_i) \le \pi(u, u_i)$, the number of iterations for $u_i$ to convert its residue entirely is bounded by $\frac{\pi(u, u_i)}{\alpha \cdot r_{max} \cdot (|N(u_i)| + |\mathcal{A}(u_i)|)}$. Each iteration for $u_i$ has a cost proportional to $(\sum_{v \in N(u_i)}{|N(v)|} + \sum_{a \in \mathcal{A}(u_i)}{|\mathcal{A}^{-1}(a)|})$. Thus, the cost for $u_i$  is $\frac{\pi(u, u_i)\cdot (\sum_{v \in N(u_i)}{|N(v)|} + \sum_{a \in \mathcal{A}(u_i)}{|\mathcal{A}^{-1}(a)|})}{\alpha \cdot r_{max} \cdot(|N(u_i)| + |\mathcal{A}(u_i)|)}$.
The total time cost for all nodes $u_i \in U$ is obtained by summing the above expression overall $u_i$. 
\end{proof}

By setting $r_{max} = \epsilon / (|E| + |\mathcal{A}|)$, as noted in \cite{wang2017fora}, FP achieves the $\epsilon$ absolute error bound. 

\stitle{Limitations of FP.}  Although FP is straightforward, it becomes slow when \( r_{\text{max}} \) is small, as this necessitates more push operations. Consequently, it struggles to efficiently compute AHPP values for large-scale graphs.

\section{Our Proposed Solution}\label{sec:proposed_algorithm} As discussed in Section \ref{sec:existing}, existing methods  must construct the state transition matrix $\mathbf{P}$ (Section \ref{subsec:ahpp}) either explicitly or implicitly. This construction is prohibitively costly due to the significant requirements for construction time and storage space (up to \(O(|U|^2)\) in the worst case). To address this challenge, in this section, we first propose a basic \underline{A}lternating \underline{P}ropagation \underline{P}ush (\algo) algorithm, which is based on two key observations regarding the Forward Push process. To further enhance computational efficiency, we introduce an advanced \underline{A}daptive \underline{S}ynchronous \underline{R}esidue \underline{P}ush (\algot) algorithm. \algot improves upon the basic method by integrating synchronous push operations and applying a more stringent residue threshold to achieve better performance. 

\begin{figure}[t!]
	\centering
	\includegraphics [width=0.4\textwidth] {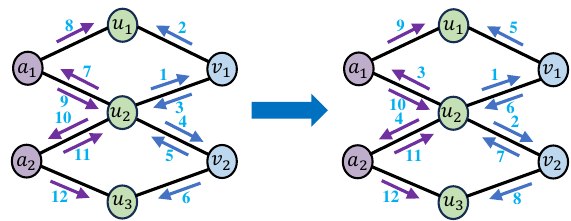}
	\caption{A sketch illustrating the transformation of the residue propagation process for node \( u_2 \), where the number on each edge represents the step of the operation (Note: This figure shows only a portion of Figure \ref{fig:preliminary_example}). }
\label{fig:APP_observation1} \vspace*{-0.3cm}
\end{figure}

\begin{figure}[t!]
	\centering
	\includegraphics [width=0.4\textwidth] {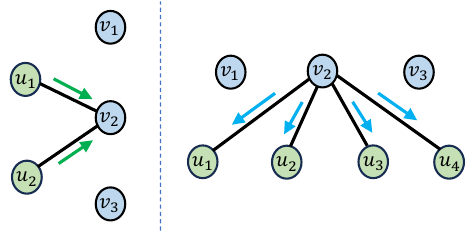}
	\caption{Efficiency bottlenecks in the Forward Push method: Residue increments from $u_1$ and $u_2$ to node \( v_2 \) are redistributed to \( v_2 \)'s neighbors twice, causing unnecessary overhead. (Note: This figure shows only a portion of Figure \ref{fig:preliminary_example}).}
\label{fig:APP_observation2} \vspace*{-0.3cm}
\end{figure}

\subsection{Basic Solution: Alternating Propagation Push} \label{subsec:APP}
\stitle{Several In-depth Observations.} (1) \textit{The equivalence of residue propagation}: in the Forward Push (i.e., Algorithm \ref{alg:FP}) method, the residue propagation process for a selected node $u_i$ can be separated into two distinct steps: (i) $u_i$ propagates its residue to each neighbor $v \in N(u_i)$ and each attribute node $a \in \mathcal{A}(u_i)$, with residues incremented by $(1 - \alpha) \cdot (1 - \beta) \cdot \frac{w(u_i, v)}{d(u_i)} \cdot r(u, u_i)$ and $(1 - \alpha) \cdot \beta \cdot \frac{w(u_i, a)}{\sum_{a_x \in \mathcal{A}(u_i)}{w(u_i, a_x)}} \cdot r(u, u_i)$, respectively, where attributes are modeled as nodes; and (ii) all non-zero residues of nodes in $V \cup \mathcal{A}$ are subsequently pushed back to $U$ losslessly. That is, node $v \in V$ (or attribute $a \in \mathcal{A}$) uniformly pushes back the residue obtained in (i) to its neighbors in $U$. An illustrative example of this observation is provided in Figure \ref{fig:APP_observation1}. Specifically, according to Algorithm \ref{alg:FP}, when node \( u_2 \) is selected, it first distributes part of its residue to node \( v_1 \), which in turn propagates its residue to all its neighbors immediately. A similar process occurs for nodes \( v_2 \), \( a_1 \), and \( a_2 \). The left part of Figure \ref{fig:APP_observation1} depicts this scenario. However, this process can be reformulated as described in the aforementioned steps. In particular, when node \( u_2 \) is selected, it propagates its residue to each of its neighbors, namely \( v_1 \) and \( v_2 \), as well as to each attribute node, \( a_1 \) and \( a_2 \). Subsequently, all residues of nodes \( v_1 \), \( v_2 \), \( a_1 \), and \( a_2 \) are pushed back to \( U \) without any loss. The right part of Figure \ref{fig:APP_observation1} illustrates this refined process. (2) \textit{Redundancy in residue propagation}:  redundancy arises when two nodes $u_i$ and $u_j$ in $U$ share a common neighbor $v_k \in V$. Specifically, during the iterations of $u_i$ and $u_j$, $v_k$ receives residues independently and then performs redundant push operations to its own neighbors, causing unnecessary computation. To illustrate, Figure \ref{fig:APP_observation2} presents an example of this observation. In Figure \ref{fig:APP_observation2}, the nodes \( u_1 \) and \( u_2 \) share a common neighbor, namely \( v_2 \). When \( u_1 \) propagates its residue to \( v_2 \), \( v_2 \) subsequently distributes the received residue to all its neighbors. Similarly, when \( u_2 \) propagates its residue to \( v_2 \), \( v_2 \) repeats the same operation. This results in redundant operations, which are unnecessary.



\begin{algorithm}[!t]
    \caption{Alternating Propagation Push (APP)}
    \label{alg:APP}
    \KwIn{Attributed bipartite graph $G$, source node $u$, restart probability $\alpha$, attribute jumping probability $\beta$, residue threshold $r_{max}$.}
    \KwOut{$\{ \hat{\pi}(u, u_i) \mid u_i \in U \}$.}
    $r(u, u) \gets 1$; $r(u, x) \gets 0 \ \forall{x \in U \cup V \cup \mathcal{A} \setminus \{u\}}$\;
    $\hat{\pi}(u, u_i) \gets 0 \ \forall{u_i \in U}$\;
    \While {$\exists u_i \in U$ s.t. $r(u, u_i)> r_{max} \cdot (|N(u_i)| + |\mathcal{A}(u_i)|)$} {
            \For {$u_i \in U$ s.t. $r(u, u_i)> r_{max} \cdot (|N(u_i)| + |\mathcal{A}(u_i)|)$} {
                $\hat{\pi}(u, u_i) \gets \hat{\pi}(u, u_i) + \alpha \cdot r(u, u_i)$\;
                $r(u, u_i) \gets (1 - \alpha) \cdot r(u, u_i)$\;
              \For {each $v_j \in N(u_i)$} {
                  $r(u, v_j) \gets r(u, v_j) + (1 - \beta) \cdot \frac{w(u_i, v_j)}{d(u_i)} \cdot r(u, u_i)$\;
              }
              \For {each $a_j \in \mathcal{A}(u_i)$} {
                  $r(u, a_j) \gets r(u, a_j) + \beta \cdot \frac{w(u_i, a_j)}{\sum_{a_x \in \mathcal{A}(u_i)}{w(u_i, a_x)}} \cdot r(u, u_i)$\;
              }
                $r(u, u_i) \gets 0$\;
            }
            \For {$v_i \in V$ s.t. $r(u, v_i) > 0$} {
                \For {each $u_j \in N(v_i)$} {
                    $r(u, u_j) \gets r(u, u_j) + \frac{w(v_i, u_j)}{d(v_i)} \cdot r(u, v_i)$\;
                }
                $r(u, v_i) \gets 0$\;
            }
            \For {$a_i \in \mathcal{A}$ s.t. $r(u, a_i) > 0$} {
                \For {each $u_j \in \mathcal{A}^{-1}(a_i)$} {
                    $r(u, u_j) \gets r(u, u_j) + \frac{w(u_j, a_i)}{\sum_{u_x \in \mathcal{A}^{-1}(a_i)}{w(u_x, a_i)}} \cdot r(u, a_i)$\;
                }
               $r(u, a_i) \gets 0$\;
            }
    }
    \Return{$\{ \hat{\pi}(u, u_i) \mid u_i \in U \}$}\;
\end{algorithm}

\stitle{Implementation Details.} Combining the aforementioned observations, \algo refines the residue distribution process, thereby enhancing computational efficiency. Specifically, Algorithm \ref{alg:APP} outlines the pseudocode for the \algo method. Taking as input an attributed bipartite graph $G$, a source node $u$, a restart probability $\alpha$, an attribute jumping probability $\beta$, and a residue threshold $r_{max}$, \algo begins by initializing the residue $r(u, u) = 1$ for the source node $u$, $r(u, x) = 0$ for all other nodes, and setting the approximate AHPP values $\hat{\pi}(u, u_i) = 0$ for all $u_i \in U$ (Lines 1-2). The algorithm then iteratively distributes the residues of selected nodes to their neighbors. For any node $u_i \in U$ with a residue exceeding $r_{max} \cdot (|N(u_i)| + |\mathcal{A}(u_i)|)$, the AHPP value $\hat{\pi}(u, u_i)$ is updated by $\alpha \cdot r(u, u_i)$ (Line 5). Its neighbor $v_j \in N(u_i)$ receives a residue increment of $(1 - \alpha) \cdot (1 - \beta) \cdot \frac{w(u_i, v_j)}{d(u_i)} \cdot r(u, u_i)$ (Lines 7-8), while each attribute node $a_j \in \mathcal{A}(u_i)$ gains a residue increment of $(1 - \alpha) \cdot \beta \cdot \frac{w(u_i, a_j)}{\sum_{a_x \in \mathcal{A}(u_i)}{w(u_i, a_x)}} \cdot r(u, u_i)$ (Lines 9-10). Once all neighbors and attribute nodes are processed, the residue $r(u, u_i)$ is reset to zero (Line 11). In the subsequent step, all non-zero residues of nodes in $V \cup \mathcal{A}$ are pushed back to $U$, ensuring no residue loss (Lines 12–19). The procedure alternates between pushing residues from $U$ to $V \cup \mathcal{A}$ and vice versa until every node $u_i \in U$ satisfies the threshold $r(u, u_i) \le r_{max} \cdot (|N(u_i)| + |\mathcal{A}(u_i)|)$, and no nodes in $V \cup \mathcal{A}$ have positive residues. This alternating residue propagation avoids materializing the transition matrix $\mathbf{P}$, significantly improving computational costs. 

\stitle{Complexity and correctness of Algorithm \ref{alg:APP}.} Algorithm \ref{alg:APP} is an adaptation of Algorithm \ref{alg:FP}, and its termination condition is identical to that of Algorithm \ref{alg:FP}. Consequently, Algorithm \ref{alg:APP} has the same absolute error and time complexity as Algorithm \ref{alg:FP}. Nonetheless, the experimental results presented in Section \ref{sec:experiments} reveal that Algorithm \ref{alg:APP} achieves a substantial improvement in computational efficiency compared to Algorithm \ref{alg:FP}.

\subsection{Advanced Solution: Adaptive Synchronous Residue Push} \label{subsec:ASRP}
\subsubsection{Important Observations and Implementation Details} We propose the \algot method based on two key observations: (1) \algo faces notable efficiency challenges in certain scenarios. For example, consider a node $v \in V$ in the attributed bipartite graph $G$, which has a large number of neighbors in the node set $U$, i.e., $|N(v)|$ is extensive. As outlined in Lines 7–8 and 13-14 of Algorithm \ref{alg:APP}, during each iteration, $v$ (i) receives residues from its selected neighbors and subsequently (ii) performs $|N(v)|$ push operations to distribute residues among its neighbors. Over multiple iterations, the majority of $v$'s neighbors may exhibit diminished residues, resulting in only a small subset being selected. Consequently, reducing a substantial portion of $v$'s residue often requires numerous iterations, with each involving at least $|N(v)|$ push operations and multiple random accesses to its neighboring nodes. (2) To meet the accuracy requirements defined in Eq. \eqref{eq:accuracy_assurance}, Algorithm \ref{alg:APP} selects nodes based on their degree and $r_{max}$, an error-threshold-independent metric, leading to unnecessary additional iterations.

Building upon the above observations, we propose \algot, which incorporates synchronous push operations and a more effective residue threshold. To address the challenge outlined in Observation 1, synchronous push operations are utilized. This technique aggregates residues from $v$'s neighbors in a single batch before redistributing them back to $v$'s neighbors, thereby streamlining the propagation process and mitigating computational overhead. However, synchronous push operations are agnostic to the specific residues associated with individual nodes, leading to a substantial number of superfluous operations that ultimately reduce overall efficiency. These limitations motivate the development of an adaptive mechanism to enhance the execution of \algo. Our proposed method adopts a greedy and adaptive approach, dynamically balancing the computational costs between \algo and synchronous push operations. Specifically, \algo is executed initially, while its actual computational cost is closely monitored. When the observed cost surpasses the estimated cost of synchronous push operations, a transition is made. This transition is regulated through a dynamically controlled threshold. Additionally, to address the limitation outlined in Observation 2, we introduce a more rigorous error threshold for \algo to improve precision and efficiency. The details of the proposed method are elaborated in Algorithm \ref{alg:ASRP}.

\begin{algorithm}[t]
		\caption{\small Adaptive Synchronous Residue Push (ASRP)}
		\label{alg:ASRP}
        \KwIn{Attributed bipartite graph $G$, source node $u$, restart probability $\alpha$, attribute jumping probability $\beta$, error threshold $\epsilon$, parameter $\lambda$}
        \KwOut{$\{ \hat{\pi}(u, u_i) \mid u_i \in U \}$}
        \nonl Lines 1-2 are the same as Lines 1-2 in Algorithm \ref{alg:APP}\;
        \setcounter{AlgoLine}{2}
        $\epsilon_f \gets \epsilon / \lambda$\;
        $n_p \gets 0$\;
        \tcc{Modified Alternating Propagation Push}
        \While {$\exists u_i \in U$ s.t. $r(u, u_i) > \epsilon_f$} {
            \For {$u_i \in U$ s.t. $r(u, u_i) > \epsilon_f$} {
                \nonl Lines 7-13 are the same as Lines 5-11 in Algorithm \ref{alg:APP}\;
                \setcounter{AlgoLine}{13}
                $n_p \gets n_p + |N(u_i)| + |\mathcal{A}(u_i)|$\;
            }
            \For {$v_i \in V$ s.t. $r(u, v_i) > 0$} {
                \nonl Lines 16-18 are the same as Lines 13-15 in Algorithm \ref{alg:APP}\;
                \setcounter{AlgoLine}{18}
                $n_p \gets n_p + |N(v_i)|$\;
            }
            \For {$a_i \in \mathcal{A}$ s.t. $r(u, a_i) > 0$} {
                \nonl Lines 21-23 are the same as Lines 17-19 in Algorithm \ref{alg:APP}\;
                \setcounter{AlgoLine}{23}
                $n_p \gets n_p + |\mathcal{A}^{-1}(a_i)|$\;
            }
            \If {$n_p\geq 2 \cdot (|E| + |E_\mathcal{A}|) \cdot \left(\log_{\frac{1}{1-\alpha}}{\frac{1}{\lambda \cdot \max_{u_i \in U}{r(u, u_i)}}}\right) $} {
                \bf{break}\;
            }
        }
        \tcc{Synchronous Push}
        \While {$\exists u_i \in U$ s.t. $r(u, u_i) > \epsilon_f$} {
            \nonl Lines 28-43 are identical to Lines 4-19 in Algorithm \ref{alg:APP} with $r_{max} \cdot (|N(u_i)| + |\mathcal{A}(u_i)|)$ replaced by 0\;
            \setcounter{AlgoLine}{43}
        }
        \Return{$\{ \hat{\pi}(u, u_i) \mid u_i \in U \}$}\;
\end{algorithm}

Specifically, Algorithm \ref{alg:ASRP} presents the pseudocode of \algot. Given an attributed bipartite graph $G$, a source node $u$, a restart probability $\alpha$, an attribute jumping probability $\beta$, an absolute error threshold $\epsilon$, and a parameter $\lambda$, \algot begins by initializing several parameters. It sets the reserve as $\hat{\pi}(u, u_i)=0$, for all $u_i \in U$ and the residue as $r(u, x) = 0$ for all $x \in U\cup V\cup \mathcal{A}$, except for $r(u, u) = 1$ (Lines 1-2). Additionally, a refined error threshold is calculated as $\epsilon_f = \epsilon / \lambda$, and the number of performed push operations is initialized to $n_p = 0$ (Lines 3-4). Subsequently, Algorithm \ref{alg:ASRP} initiates the asynchronous push process, following the procedure outlined in Lines 3–19 of Algorithm \ref{alg:APP}, with a key modification: it replaces the residue threshold $r_{max} \cdot (|N(u_i)| + |\mathcal{A}(u_i)|)$ with $\epsilon_f$. During this process, the recorded cost $n_p$ of the modified Alternating Propagation Push is incremented by $(|N(u_i)| + |\mathcal{A}(u_i)|)$ (or $|N(v_i)|$, $|\mathcal{A}^{-1}(a_i)|$) whenever the neighboring nodes of $u_i \in U$ (or $v_i \in V$, $a_i \in \mathcal{A}$) are accessed (Lines 14, 19, and 24).
Furthermore, at Lines 25-26, if Algorithm \ref{alg:ASRP} exhausts its computational budget for asynchronous pushes before meeting the termination criteria outlined in Line 5, it strategically transitions to iterative synchronous pushes. In each iteration of synchronous pushes, \algot performs push operations for all nodes with positive residues, as described in Lines 3-19 of Algorithm \ref{alg:APP}, with the exception that the residue threshold is replaced by zero. Algorithm \ref{alg:ASRP} terminates and returns $\{ \hat{\pi}(u, u_i) \mid u_i \in U \}$ when $r(u, u_i) \leq \epsilon_f$ for all $u_i \in U$.

\subsubsection{Theoretical Analysis} Next, we analyze the time complexity and correctness of Algorithm \ref{alg:ASRP}.

\begin{theorem}[Time Complexity]\label{thm:ASRP_Time}
The time complexity of Algorithm \ref{alg:ASRP} is evaluated as $O\left((|E| + |E_\mathcal{A}|)\cdot \log(1/\epsilon)\right)$.
\end{theorem}

\begin{proof}
If Algorithm \ref{alg:ASRP} terminates upon depleting its allocated computation budget for asynchronous pushes, the cost incurred by this phase is hence $ O((|E| + |E_\mathcal{A}|) \cdot (\log{\frac{1}{\lambda \cdot \max_{u_i \in U}{r(u, u_i)}}})).$ During the synchronous push phase, each iteration reduces a portion $\alpha$ of $r(u, u_i)$ for all $u_i \in U$, converting it into approximate AHPP values. Let $t$ denote the number of iterations required for synchronous pushes. The process halts when $\forall u_i \in U$, $r^{(t)}(u, u_i) \leq \epsilon_f$, where $r^{(t)}(u, u_i)$ represents the residue after $t$ iterations. This condition gives the following inequality:
\begin{align*}
\textstyle (1 - \alpha \cdot \sum_{\ell = 0}^{t - 1}{(1 - \alpha)^{\ell}}) \cdot \max_{u_i \in U}{r(u, u_i)} \le \epsilon_f.
\end{align*}
Solving for $t$, we have $t = \log_{\frac{1}{1 - \alpha}} ( \frac{\max_{u_i \in U} r(u, u_i)}{\epsilon_f})$. Each iteration (Lines 27-43) of synchronous pushes involves traversing the graph; thus, the time cost per iteration is bounded by $O(|E| + |E_\mathcal{A}|)$. Thus, the overall time complexity of Algorithm \ref{alg:ASRP} is $O\left(n_p + (|E| + |E_\mathcal{A}|) \cdot t \right)$, which simplifies to $O((|E| + |E_\mathcal{A}|) \cdot \log(\frac{1}{\epsilon}))$.
\end{proof}

\begin{lemma}[Iterative invariance]\label{lem:fwd_equation} At the termination of any iteration in Algorithm \ref{alg:ASRP}, the following equation holds.
\begin{equation}\label{eq:fwd_equation}
    \pi(u, u_i) = \hat{\pi}(u, u_i) + \sum_{u_j \in U}{r(u, u_j) \cdot \pi(u_j, u_i)}
\end{equation}
\end{lemma}

\begin{proof}
We prove this lemma using mathematical induction. First, we consider the case where the residue threshold is $\epsilon_f$. Initially, we have $r(u, u_i) = 0$ for all $u_i \in U$, except for $r(u, u) = 1$, and $\hat{\pi}(u, u_i) = 0$ for each node $u_i \in U$. Therefore, we have $\pi(u, u_i) = r(u, u) \cdot \pi(u, u_i)$ for all $u_i \in U$, implying that Eq. \eqref{eq:fwd_equation} holds in the base case. Next, for the inductive case, we assume that after $\ell$ iterations, the approximate AHPP values $\hat{\pi}^{(\ell)}(u, u_i)$ for all $u_i \in U$, and the residue values $r^{(\ell)}(u, u_i)$ for each node $u_i \in U$, satisfy Eq. \eqref{eq:fwd_equation}, i.e., $\pi(u, u_i) = \hat{\pi}^{(\ell)}(u, u_i) + \sum_{u_j \in U}{r^{(\ell)}(u, u_j) \cdot \pi(u_j, u_i)}$.
Let $\hat{\pi}^{(\ell+1)}(u, u_i)$ for all $u_i \in U$ be the approximate AHPP values, and $r^{(\ell+1)}(u, u_i)$ be the residue values at the end of the $(\ell+1)$-th iteration. For any node $u_i\in U$, we define $\Delta(u_i)$:
\begin{align}\label{eq:delta_ui}
& \Delta(u_i) = \hat{\pi}^{(\ell+1)}(u, u_i) - \hat{\pi}^{(\ell)}(u, u_i) \nonumber\\
&\textstyle + \sum_{u_j \in U}{\left(r^{(\ell+1)}(u, u_j) - r^{(\ell)}(u, u_j)\right) \cdot \pi(u_j, u_i)}.
\end{align}
If $\Delta(u_i) = 0$ holds for each $u_i\in U$, the lemma is established. By Algorithm \ref{alg:ASRP}, for a node $u_i$ with $r^{(\ell)}(u, u_i) > \epsilon_f$, we have $\hat{\pi}^{(\ell+1)}(u, u_i) - \hat{\pi}^{(\ell)}(u, u_i) = \alpha \cdot r^{(\ell)}(u, u_i)$ and $r^{(\ell+1)}(u, u_i) - r^{(\ell)}(u, u_i) = -r^{(\ell)}(u, u_i) + \sum_{\substack{u_x \in U \\ r^{(\ell)}(u, u_x) > \epsilon_f}}{(1-\alpha) \cdot r^{(\ell)}(u, u_x) \cdot \mathbf{P}(u_x, u_i)}.$
Whereas for $u_i$ with $r^{(\ell)}(u, u_i) \le \epsilon_f$, we have $\hat{\pi}^{(\ell+1)}(u, u_i) - \hat{\pi}^{(\ell)}(u, u_i) = 0$ and
$r^{(\ell+1)}(u, u_i) - r^{(\ell)}(u, u_i)= \sum_{\substack{u_x \in U \\ r^{(\ell)}(u, u_x) > \epsilon_f}}{(1 - \alpha) \cdot r^{(\ell)}(u, u_x) \cdot \mathbf{P}(u_x, u_i)}$.
If $r^{(\ell)}(u, u_i) > \epsilon_f$, Eq. \eqref{eq:delta_ui} can be transformed as follows.
\begin{align*}
&\scriptstyle \Delta(u_i) = \alpha \cdot r^{(\ell)}(u, u_i) - \sum_{\substack{u_j \in U \\ r^{(\ell)}(u, u_j) > \epsilon_f}}{r^{(\ell)}(u, u_j) \cdot \pi(u_j, u_i)} \\
&\scriptstyle + \sum_{u_j \in U}{\left(\sum_{\substack{u_x \in U \\ r^{(\ell)}(u, u_x) > \epsilon_f}}{(1 - \alpha) \cdot r^{(\ell)}(u, u_x) \cdot \mathbf{P}(u_x, u_j)}\right) \cdot \pi(u_j, u_i)}.
\end{align*}
Simplifying, we find the following relationship.
\begin{align*}
&\scriptstyle \Delta(u_i) = \alpha \cdot r^{(\ell)}(u, u_i) - \sum_{\substack{u_j \in U \\ r^{(\ell)}(u, u_j) > \epsilon_f}}{r^{(\ell)}(u, u_j) \cdot \pi(u_j, u_i)} \\
&\scriptstyle + \sum_{\substack{u_x \in U \\ r^{(\ell)}(u, u_x) > \epsilon_f}}{\sum_{u_j \in U}{(1 - \alpha) \cdot r^{(\ell)}(u, u_x) \cdot \mathbf{P}(u_x, u_j) \cdot \pi(u_j, u_i)}} \\
&\scriptstyle = \alpha \cdot r^{(\ell)}(u, u_i) + \sum_{u_j \in U}{(1 - \alpha) \cdot r^{(\ell)}(u, u_i) \cdot \mathbf{P}(u_i, u_j) \cdot \pi(u_j, u_i)} \\
&\scriptstyle - r^{(\ell)}(u, u_i) \cdot \pi(u_i, u_i) - \sum_{\substack{u_j \in U \\ r^{(\ell)}(u, u_j) > \epsilon_f \\ u_j \ne u_i}}{r^{(\ell)}(u, u_j) \cdot \pi(u_j, u_i)} \\
&\scriptstyle + \sum_{\substack{u_x \in U \\ r^{(\ell)}(u, u_x) > \epsilon_f \\ u_x \ne u_i}}{\sum_{u_j \in U}{(1 - \alpha) \cdot r^{(\ell)}(u, u_x) \cdot \mathbf{P}(u_x, u_j) \cdot \pi(u_j, u_i)}} = 0.
\end{align*}
If $r^{(\ell)}(u, u_i) \le \epsilon_f$, we have
\begin{align*}
&\textstyle \Delta(u_i) = - \sum_{\substack{u_j \in U \\ r^{(\ell)}(u, u_j) > \epsilon_f}}{r^{(\ell)}(u, u_j) \cdot \pi(u_j, u_i)} \\
&\scriptstyle + \sum_{u_j \in U}{\left(\sum_{\substack{u_x \in U \\ r^{(\ell)}(u, u_x) > \epsilon_f}}{(1 - \alpha) \cdot r^{(\ell)}(u, u_x) \cdot \mathbf{P}(u_x, u_j)}\right) \cdot \pi(u_j, u_i)} \\
&\scriptstyle = \sum_{\substack{u_x \in U \\ r^{(\ell)}(u, u_x) > \epsilon_f}}{\sum_{u_j \in U}{(1 - \alpha) \cdot r^{(\ell)}(u, u_x) \cdot \mathbf{P}(u_x, u_j) \cdot \pi(u_j, u_i)}} \\
&\textstyle - \sum_{\substack{u_j \in U \\ r^{(\ell)}(u, u_j) > \epsilon_f}}{r^{(\ell)}(u, u_j) \cdot \pi(u_j, u_i)} = 0.
\end{align*}
Thus, the proof is established when the residue threshold is $\epsilon_f$. The proof for the case where the residue threshold is $0$ follows similarly and is omitted for brevity.
\end{proof}


\begin{theorem}[Correctness]\label{thm:ASRP_Correctness}
Given a source node $u$, \algot computes an approximate AHPP value $\hat{\pi}(u, u_i)$ for each node $u_i \in U$ with $\lambda \ge \max_{u_i \in U}{\pi(u, u_i)}$, ensuring that:
\begin{equation*}
    |\pi(u, u_i) - \hat{\pi}(u, u_i)| \le \epsilon.
\end{equation*}
\end{theorem}

\begin{proof}
At the termination of Algorithm \ref{alg:ASRP}, $ r(u, u_i) \le \epsilon_f$ holds for every node $u_i \in U$. By applying Lemma \ref{lem:fwd_equation}, we have $\pi(u, u_i) - \hat{\pi}(u, u_i) \textstyle= \sum_{u_j \in U}{r(u, u_j) \cdot \pi(u_j, u_i)}  \le \epsilon_f \cdot \sum_{u_j \in U}{\pi(u_j, u_i)}  \le \epsilon_f \cdot \lambda = \epsilon$.
Besides, from Lemma \ref{lem:fwd_equation}, we also have $\pi(u, u_i) - \hat{\pi}(u, u_i) = \sum_{u_j \in U}{r(u, u_j) \cdot \pi(u_j, u_i)} \ge 0$.
Combining the two inequalities, we obtain $|\pi(u, u_i) - \hat{\pi}(u, u_i)| \le \epsilon,$ which completes the proof.
\end{proof}


\stitle{The estimation of $\lambda$.} By Theorem \ref{thm:ASRP_Correctness}, we know that the input parameter $\lambda$ of Algorithm \ref{alg:ASRP} must be greater than $\max_{u_i \in U} \pi(u, u_i)$. However, $\pi(u, u_i)$ is the output of Algorithm \ref{alg:ASRP}. Thus, at first glance, it seems very difficult for us to set $\lambda$, but the following Theorem \ref{thm:lambda_estimation}  shows that we can effectively estimate $\lambda$ during the pre-processing phase through Equation \eqref{eq:lambda_estimation}. Moreover, Theorem \ref{thm:lambda_estimation} guarantees that this estimate constitutes a tight upper bound for $\lambda$.

\begin{theorem} \label{thm:lambda_estimation}
Let $\boldsymbol{\pi}_T$ denote the result of PI initialized with the all-ones vector $\boldsymbol{1}$ after $T$ iterations. Then the following inequality holds:
$$
\lambda \geq \max_{u_i \in U} \sum_{u_j \in U} \pi(u_j, u_i),
$$
where $\lambda$ is defined as:
\begin{equation} \label{eq:lambda_estimation}
\lambda = \max_{u_i \in U}{\boldsymbol{\pi}_T(u_i)} + |U| \cdot (1 - \alpha)^T.
\end{equation}
\end{theorem}

\begin{proof}
Following the PI process initialized with the all-ones vector $\boldsymbol{1}$ and run for $T$ iterations, we obtain:
\begin{align*}
\boldsymbol{\pi}_T = \boldsymbol{1} \cdot \sum_{\ell=0}^{T-1} \alpha(1-\alpha)^{\ell} \mathbf{P}^{\ell}.
\end{align*}

For any node $u_i \in U$, the cumulative similarity satisfies:
\begin{align*}
\sum_{u_j \in U} \pi(u_j, u_i) &= \boldsymbol{\pi}_T(u_i) + \left(\boldsymbol{1} \cdot \sum_{\ell=T}^{\infty} \alpha(1-\alpha)^{\ell} \mathbf{P}^{\ell}\right)(u_i) \\
&\leq \boldsymbol{\pi}_T(u_i) + |U| \cdot \left(\sum_{\ell=T}^{\infty} \alpha(1-\alpha)^{\ell}\right) \\
&= \boldsymbol{\pi}_T(u_i) + |U| \cdot (1-\alpha)^T \\
&\leq \max_{u_i \in U} \boldsymbol{\pi}_T(u_i) + |U| \cdot (1-\alpha)^T,
\end{align*}
where the first inequality follows from the fact that each entry of $\mathbf{P}^{\ell}$ is bounded by 1. This completes the proof.
\end{proof}

The PI used in Theorem \ref{thm:lambda_estimation} is a variant of the original PI (Section \ref{subsec:pi}), enhanced with synchronous push operations. This adaptation, as demonstrated empirically, optimizes the trade-off between time efficiency and solution quality.

\section{Experimental Evaluation}\label{sec:experiments}

   

\begin{table}[t!]
	\centering
	\caption{\textcolor{blue}{Dataset Statistics. ($K = 10^3$, $M = 10^6$)}} 
	\scalebox{1}{
		\begin{tabular}{c|c|c|c|c|c}
			\toprule
			\textbf{Dataset} & $\boldsymbol{|U|}$ & $\boldsymbol{|V|}$ & $\boldsymbol{|E|}$ & $\boldsymbol{|\mathcal{A}|}$ & $\boldsymbol{|E_{\mathcal{A}}|}$\\
			\toprule
            \textcolor{blue}{Cora \cite{yang2024effective}} & 1.3K & 0.7K & 2.3K & 1.4K & 23.8K\\
			Citeseer \cite{yang2024effective} & 1.2K & 0.7K & 1.6K & 3.7K & 38.8K\\
            \textcolor{blue}{BlogCatalog \cite{DBLP:conf/wsdm/MengLBZ19}} & 3.4K & 1.7K & 83.0K & 8.1K & 245.2K\\
			Flickr \cite{yang2021effective} & 3.7K & 3.2K & 93.0K & 12.0K & 93.7K\\
            MovieLens  \cite{yang2024effective} & 6.0K & 3.8K & 1.0M & 0.03K & 18.1K\\ 
			LastFM \cite{yang2022scalable} & 7.6K & 7.8K & 3.0M & 7.6K & 55.6K\\
            Google \cite{yang2024effective} & 64.5K & 800.0K & 1.4M & 1.0K & 9.6M\\
			Amazon \cite{yang2024effective} & 2.3M & 8.0M & 22.5M & 0.8K & 24.2M\\
			\bottomrule		
	\end{tabular}}
	\label{tab:data}
\end{table}

In this section, we conduct comprehensive experiments to evaluate the effectiveness and efficiency of our solutions. All algorithms are implemented in C++ and compiled with g++ 11.4.0 using -O3 optimization. All experiments are conducted on a Linux machine with an Intel Xeon(R) Silver 4210 @2.20GHz CPU and 1 TB RAM.



\vspace{-0.2cm}
\subsection{Experimental Setup}
\stitle{Datasets.} \textcolor{blue}{We evaluate our methods on eight widely used attributed bipartite graph benchmarks for similarity search (Table~\ref{tab:data}) \cite{huang2020biane, yang2021effective, yang2024effective, DBLP:conf/wsdm/MengLBZ19, yang2022scalable, yang2022efficient}. Cora and CiteSeer are citation networks with publications as nodes and citations as edges. BlogCatalog is a social network of bloggers, where node attributes are derived from user-provided keywords describing blogs. Flickr models user--tag interactions, while LastFM represents user--artist preferences. MovieLens is a user-movie network in which edges correspond to user ratings. The Google and Amazon datasets capture user-generated reviews of restaurants and books, respectively. All datasets are equipped with ground-truth cluster labels, following the experimental protocols adopted in prior work \cite{huang2020biane, yang2021effective, yang2024effective, DBLP:conf/wsdm/MengLBZ19, yang2022scalable, yang2022efficient}.}


\begin{figure*}[t!]
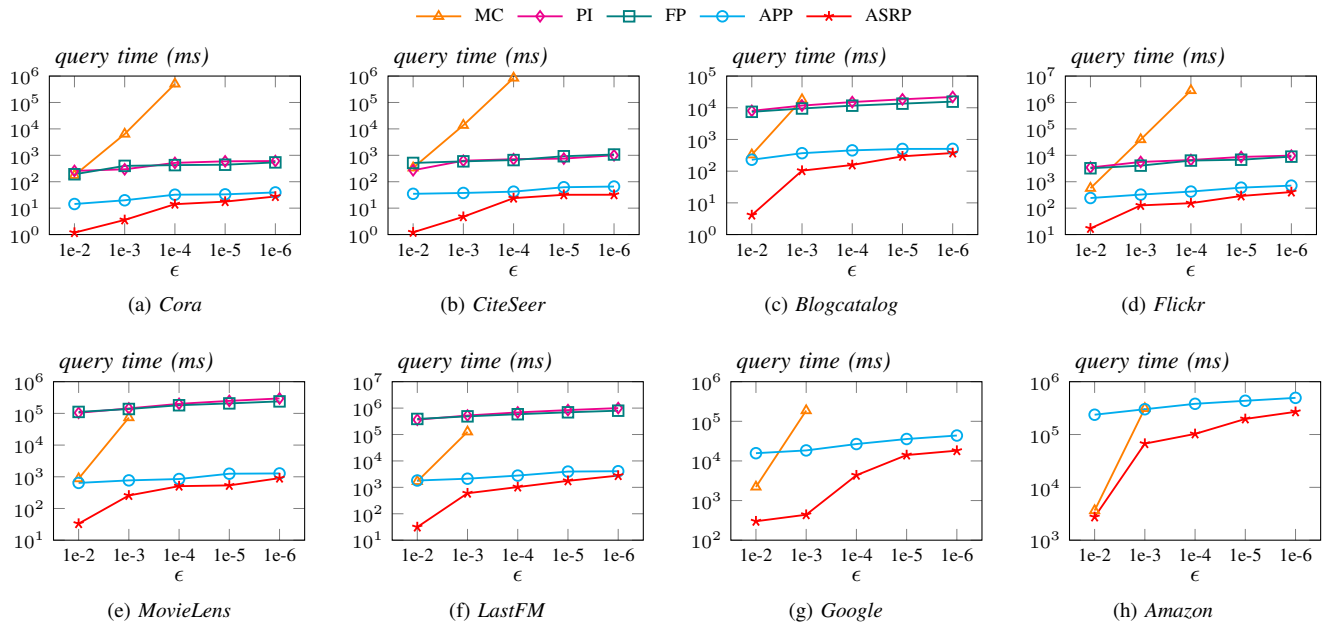

\centering
\begin{small}
%
}%
\end{small}
\caption{\textcolor{blue}{Running time of different methods with varying $\epsilon$.}}
\label{fig:time_epsilon}
\vspace{-3mm}
\end{figure*}

\begin{figure*}[t!]
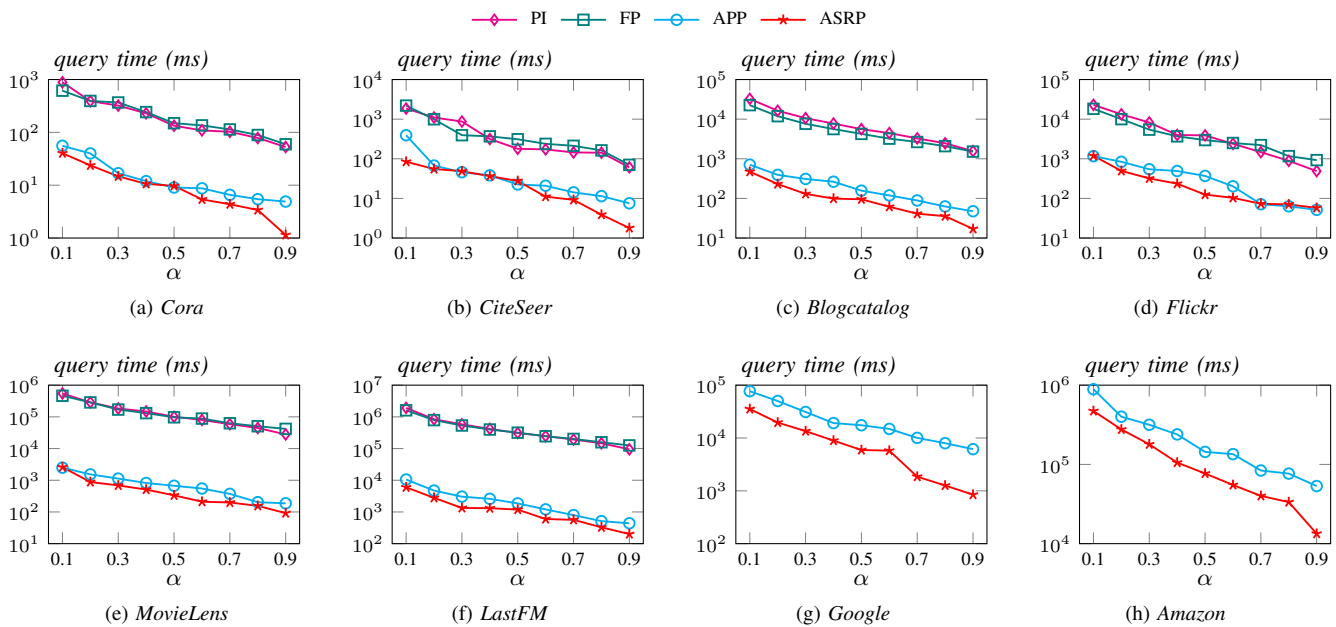

\centering
\begin{small}
%
\hspace{0mm}%
}%
\end{small}
\caption{\textcolor{blue}{Running time of different methods with varying $\alpha$ (MC is excluded due to timeouts on all datasets).}}
\label{fig:time_alpha}
\vspace{-3mm}
\end{figure*}

\stitle{Baselines.} \textcolor{blue}{We compare our approach with three categories of representative methods. (1) Unipartite graph similarity search methods: Pearson \cite{huang2024zhenmei}, Jaccard \cite{li2024improved}, SimRank \cite{wang2021exactsim}, PPR \cite{DBLP:journals/pacmmod/ZhouLLLW25}, and RW\_Uniform \cite{yin2010unified}. (2) Bipartite graph similarity search methods: AnchorGNN \cite{wu2023billion}, GEBEp \cite{yang2022scalable}, IDBR \cite{DBLP:conf/sigir/LiWLLJD24}, BiNE \cite{DBLP:journals/tkde/GaoHCLZZ22}, HPP \cite{yang2022efficient}, and BHPP \cite{yang2022efficient}. (3) Attributed bipartite graph similarity search methods: BiANE \cite{huang2020biane}, together with MC, PI, FP, APP, and ASRP, which are different algorithmic solutions to the proposed AHPP problem (Sections~\ref{subsec:mc}--\ref{subsec:ASRP}).}


\stitle{Parameters.} The proposed AHPP model involves two parameters: the restart probability $\alpha$ and the attribute jumping probability $\beta$. Unless otherwise specified,  we set $\alpha=0.15$ and $\beta=0.35$ by default. For a fair comparison, the parameters for other baselines are set according to the recommendations in their respective papers. For reliability, we randomly select 100 query nodes and report the average running time and quality.


\subsection{Efficiency Evaluation} \label{subsec:efficiency}

\textcolor{blue}{This subsection compares the running times of AHPP algorithms (MC, PI, FP, APP, and ASRP) only. Other methods are excluded because their objectives differ, making efficiency comparisons meaningless. For brevity, we also omit methods whose average query time exceeds one week.}


\begin{figure*}[t!]
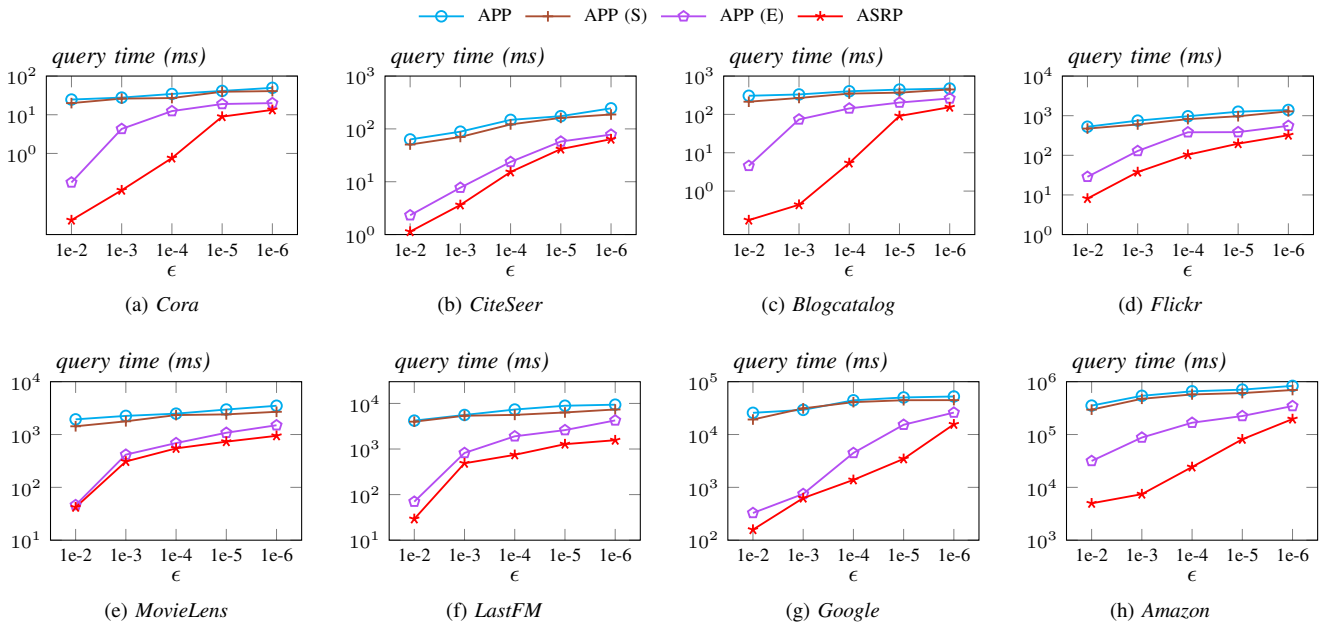

\centering
\begin{small}
\hspace{0mm}%
}%
\end{small}
\caption{\textcolor{blue}{The performance of different optimizers of ASRP with varying $\epsilon$.}}
\label{fig:opt_epsilon}
\vspace{-3mm}
\end{figure*}

\begin{figure*}[t!]
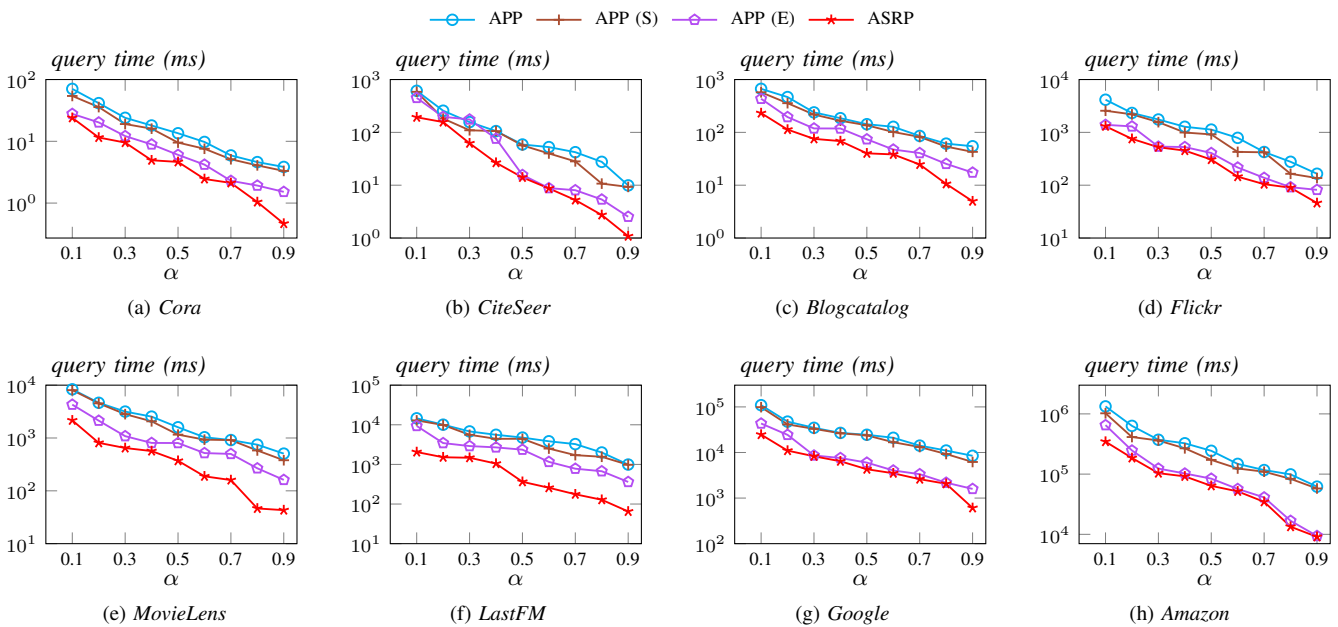

\centering
\begin{small}
%
}%
\end{small}
\caption{\textcolor{blue}{The performance of different optimizers of ASRP with varying $\alpha$.}}
\label{fig:opt_alpha}
\vspace{-3mm}
\end{figure*}

\stitle{Exp-1: Running time of different methods with varying parameters.} \textcolor{blue}{For MC, we set \( p_f = 10^{-6} \) following previous works \cite{wang2017fora, liao2022efficient, liao2023efficient}. For ASRP, we set \( \lambda = \max_{u_i \in U}{\pi(u, u_i)} \), as stated in Theorem \ref{thm:ASRP_Correctness}. Additionally, there are three parameters \( \epsilon \), \( \alpha \), and \( \beta \) are involved in MC, PI, FP, APP, and ASRP (details in  Sections \ref{sec:existing} and \ref{sec:proposed_algorithm}). Since \( \beta \) affects only query quality and not running time, we omit the results for different values of \( \beta \) for brevity. For the parameter \(\epsilon\), by Figure \ref{fig:time_epsilon}, we have the following observations: (1) The proposed APP and ASRP are at least an order of magnitude faster than existing methods in most cases. Notably, ASRP is at least two orders of magnitude faster than the three baselines (i.e., MC, PI, FP) across all datasets for \(10^{-3} \leq \epsilon \leq 10^{-6}\), highlighting the high efficiency of  ASRP, which is consistent with the theoretical analysis in  Section \ref{subsec:ASRP}. (2) The MC approach struggles with high-precision queries (e.g., \(\epsilon \in \{10^{-5}, 10^{-6}\}\)) and is inefficient in such cases, while PI and FP fail to scale effectively on large graphs (e.g., Google and Amazon). (3) The running time of all methods increases as precision increases (smaller \(\epsilon\)), as more iterations are required.  In particular, MC  experiences a significant increase, as higher precision necessitates more random walks. For the parameter $\alpha$,  we fix $\epsilon=10^{-6}$. By Figure \ref{fig:time_alpha}, we observe the following: (1) The proposed APP and ASRP exhibit the highest efficiency on all datasets, frequently outperforming baseline methods by 1$\sim$2 orders of magnitude. In particular, compared to the two competitors, PI and FP, ASRP achieves a speedup of at least 319\(\times\) on MovieLens, LastFM, Google, and Amazon, and over 22\(\times\) on Cora, CiteSeer, BlogCatalog, and Flickr. (2) PI and FP have high computational costs on large-scale graphs, such as Google and Amazon. (3) The computation time for all methods decreases as \(\alpha\) increases, as larger \(\alpha\) values lead to fewer iterations, which aligns with our theoretical analysis (details in  Sections \ref{sec:existing} and \ref{sec:proposed_algorithm}). In short, these results provide evidence that the proposed APP and ASRP are highly efficient when contrasted with baselines.}

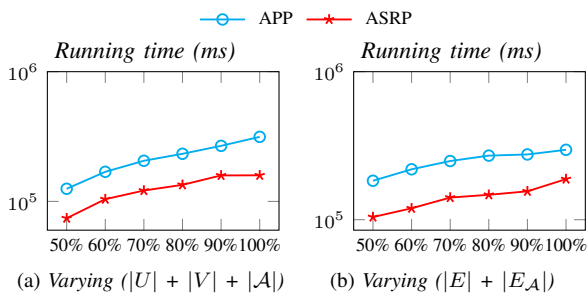
\begin{figure}[t!]
\centering
\begin{small}
\begin{tikzpicture}
    \begin{customlegend}[legend columns=2,
        legend entries={APP, ASRP},
        legend style={at={(0.45,1.05)},anchor=north,draw=none,font=\scriptsize,column sep=0.1cm}]
    \addlegendimage{line width=0.25mm,mark=o,color=first-color}
    \addlegendimage{line width=0.25mm,mark=star,color=second-color}
    \end{customlegend}
\end{tikzpicture}
\\[-\lineskip]
\vspace{-4mm}
\subfloat[{\em Varying ($|U|$ + $|V|$ + $|\mathcal{A}|$)}] {
\hspace{-5mm}
\begin{tikzpicture}[scale=1]
    \begin{axis}[
        height=\columnwidth/2.4,
        width=\columnwidth/1.9,
        ylabel={\em Running time (ms)},
        xmin=0.5, xmax=6.5,
        ymin=60000, ymax=1000000,
        xtick={1,2,3,4,5,6},
        xticklabel style = {font=\scriptsize},
        yticklabel style = {font=\scriptsize},
        xticklabels={50\%, 60\%, 70\%, 80\%, 90\%, 100\%},
        scaled y ticks = false,
        ymode=log,
        log basis y={10},
        ytick = {10000, 100000, 1000000},
        every axis y label/.style={at={(current axis.north west)},right=13mm,above=0mm},
    ]
    \addplot[line width=0.25mm,mark=o,color=first-color] 
            plot coordinates {
(1, 125115)
(2, 169069)
(3, 205656)
(4, 232806)
(5, 268066)
(6, 314415)
    };

    \addplot[line width=0.25mm,mark=star,color=second-color] 
    plot coordinates {
(1, 74160)
(2, 103953)
(3, 121266)
(4, 134153)
(5, 158562)
(6, 158896)
    };
    \end{axis}
\end{tikzpicture}\hspace{0mm}%
}
\subfloat[{\em Varying ($|E|$ + $|E_\mathcal{A}|$)}] {
\begin{tikzpicture}[scale=1]
    \begin{axis}[
        height=\columnwidth/2.4,
        width=\columnwidth/1.9,
        ylabel={\em Running time (ms)},
        xmin=0.5, xmax=6.5,
        ymin=85000, ymax=1000000,
        xtick={1,2,3,4,5,6},
        xticklabel style = {font=\scriptsize},
        yticklabel style = {font=\scriptsize},
        xticklabels={50\%, 60\%, 70\%, 80\%, 90\%, 100\%},
        scaled y ticks = false,
        ymode=log,
        log basis y={10},
        ytick = {100000, 1000000},
        every axis y label/.style={at={(current axis.north west)},right=13mm,above=0mm},
    ]
    \addplot[line width=0.25mm,mark=o,color=first-color] 
            plot coordinates {
(1, 183302)
(2, 219053)
(3, 249119)
(4, 270978)
(5, 275994)
(6, 296619)
    };

    \addplot[line width=0.25mm,mark=star,color=second-color] 
    plot coordinates {
(1, 104245)
(2, 119450)
(3, 141123)
(4, 147349)
(5, 155562)
(6, 188102)
    };
    \end{axis}
\end{tikzpicture}\hspace{0mm}%
}
\vspace{0mm}
\end{small}
\caption{Scalability testing on synthetic graphs.}
\label{fig:sca}
\end{figure}

\stitle{Exp-2: The performance of different optimizers of ASRP with varying parameters}. \textcolor{blue}{In this experiment, we assess the impact of two optimizations (introduced in Section \ref{subsec:ASRP}) on the \algo algorithm under varying values of $\epsilon$ and $\alpha$. For the sake of clarity, we denote the \algo algorithm with synchronous push operations and a more stringent error threshold as \algo (S) and \algo (E), respectively. The parameter settings are identical to those in Exp-1. The results consistently indicate a reduction in runtime with the addition of each optimization. Specifically, based on Figures \ref{fig:opt_epsilon} and \ref{fig:opt_alpha}, the following observations can be made. (1) Both \algo (S) and \algo (E) outperform \algo in all datasets, highlighting the efficacy of the proposed optimizations. (2) The fully optimized algorithm (i.e., \algot) demonstrates the highest efficiency in all datasets, suggesting that both optimizations are effective. (3) The \algo (E) algorithm is faster than the \algo (S) algorithm in all instances, implying that the optimization with a stringent error threshold might be more effective than the one with synchronous push operations. In conclusion, the fully optimized version \algot yields the best results, and these findings emphasize the effectiveness of the proposed optimizations.}


\stitle{Exp-3: Scalability testing on synthetic graphs.} \textcolor{blue}{To further evaluate the scalability of our proposed solutions, we construct a comprehensive set of synthetic graphs by sampling nodes or edges together with their attributes from the Amazon dataset, which contains the largest numbers of nodes and edges among all datasets used in our experiments. Specifically, the sampling ratio is varied from 50\% to 100\%. Figure~\ref{fig:sca} reports the experimental results on these synthetic graphs with $\epsilon = 10^{-6}$, and similar performance trends are observed for other choices of $\epsilon$. The results of MC, PI, and FP are omitted, as these methods fail to terminate within the predefined time limit of one week. This observation is further supported by the results shown in Figures~\ref{fig:time_epsilon} and~\ref{fig:time_alpha}. As illustrated, the proposed APP and ASRP algorithms exhibit near-linear scalability with respect to the size of the synthetic graphs. Moreover, ASRP consistently outperforms APP in terms of running time. These findings demonstrate that our algorithms are capable of handling large-scale attributed bipartite graphs, which is consistent with our theoretical analysis presented in Section~\ref{sec:proposed_algorithm}.}



\subsection{Effectiveness Evaluation}
We use the well-known clustering consistency validation \cite{yang2021effective, yang2024effective}, top-k precision \cite{wang2020personalized, wang2017fora}, and link prediction \cite{huang2020biane, DBLP:conf/sigir/LiWLLJD24} to evaluate the effectiveness of the proposed solutions. Detailed results are reported as follows.


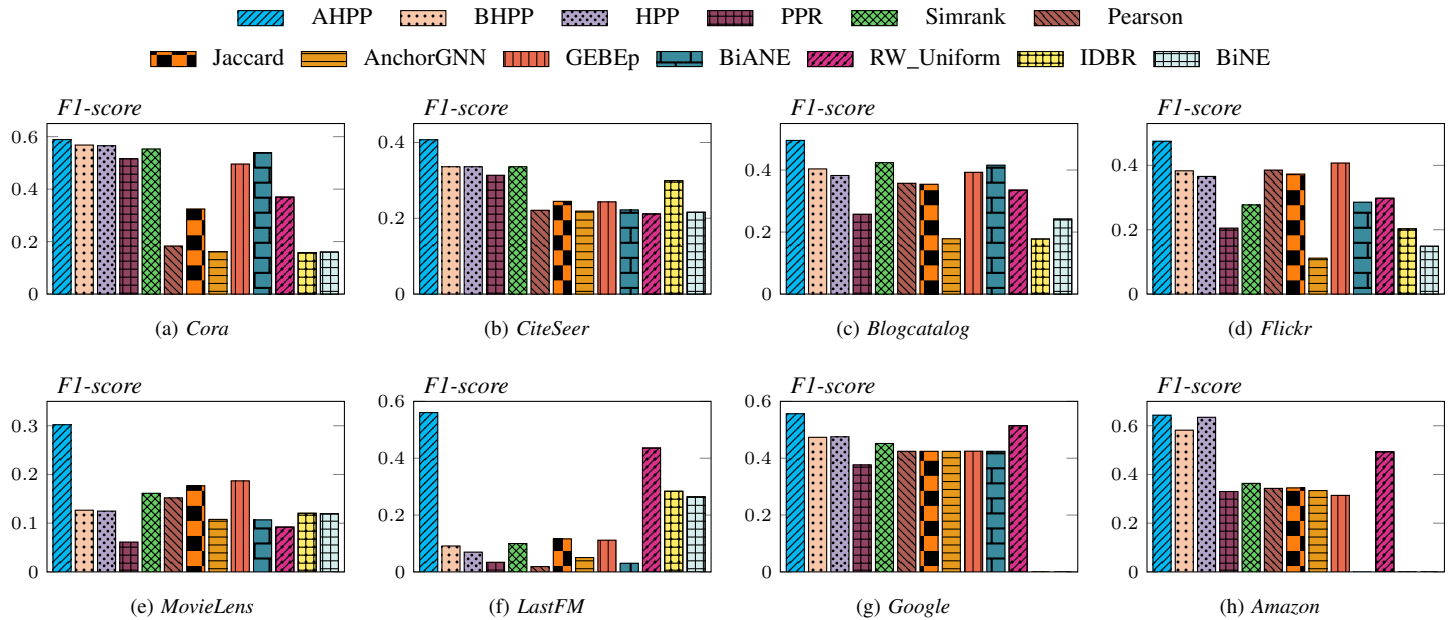
\begin{figure*}[t!]
\centering
\begin{small}
\begin{tikzpicture}
    \begin{customlegend}[
        legend entries={AHPP, BHPP, HPP, PPR, Simrank, Pearson},
        legend columns=6,
        area legend,
        legend style={at={(0.5,1.15)},anchor=north,draw=none,font=\small,column sep=0.3cm}]
        \addlegendimage{preaction={fill, AHPP-color}, pattern=north east lines} 
        \addlegendimage{preaction={fill, BHPP-color}, pattern=dots} 
        \addlegendimage{preaction={fill, HPP-color}, pattern=crosshatch dots} 
        \addlegendimage{preaction={fill, PPR-color}, pattern=grid} 
        \addlegendimage{preaction={fill, Simrank-color}, pattern=crosshatch} 
        \addlegendimage{preaction={fill, Pearson-color}, pattern=north west lines} 
    \end{customlegend}
\end{tikzpicture}
\begin{tikzpicture}
    \begin{customlegend}[
        legend entries={Jaccard, AnchorGNN, GEBEp, BiANE, RW\_Uniform, IDBR, BiNE},
        legend columns=7,
        area legend,
        legend style={at={(0.5,1.15)},anchor=north,draw=none,font=\small,column sep=0.15cm}]
        \addlegendimage{preaction={fill, Jaccard-color}, pattern=checkerboard} 
        \addlegendimage{preaction={fill, AnchorGNN-color}, pattern=horizontal lines} 
        \addlegendimage{preaction={fill, GEBEp-color}, pattern=vertical lines} 
        \addlegendimage{preaction={fill, BiANE-color}, pattern=bricks} 
        \addlegendimage{preaction={fill, RW_Uniform-color}, pattern=north east lines, postaction={pattern=dots}} 
        \addlegendimage{preaction={fill, IDBR-color}, pattern=grid, postaction={pattern=dots}} 
        \addlegendimage{preaction={fill, BiNE-color}, pattern=horizontal lines, postaction={pattern=vertical lines}} 
    \end{customlegend}
\end{tikzpicture}
\\[-\lineskip]
\vspace{-2mm}
\subfloat[{\em Cora}]{
\hspace{-7mm}
\begin{tikzpicture}[scale=1]
\begin{axis}[
        height=\columnwidth/2.3,
        width=\columnwidth/1.6,
    xtick=\empty,
    ybar=1.6pt,
    bar width=0.24cm,
    enlarge x limits=true,
    ylabel={\em F1-score},
    xticklabel=\empty,
    yticklabel style = {font=\scriptsize},
    ymin=0,
    ymax=0.65,
    every axis y label/.style={at={(current axis.north west)},right=7mm,above=0mm},
    legend style={at={(0.02,0.98)},anchor=north west,cells={anchor=west},font=\tiny}
    ]

\addplot [preaction={fill, AHPP-color}, pattern=north east lines] coordinates {(1, 0.58883149) }; 
\addplot [preaction={fill, BHPP-color}, pattern=dots] coordinates {(1, 0.56786829) }; 
\addplot [preaction={fill, HPP-color}, pattern=crosshatch dots] coordinates {(1, 0.56535561) }; 
\addplot [preaction={fill, PPR-color}, pattern=grid] coordinates {(1, 0.51489750) }; 
\addplot [preaction={fill, Simrank-color}, pattern=crosshatch] coordinates {(1, 0.55294288) }; 
\addplot [preaction={fill, Pearson-color}, pattern=north west lines] coordinates {(1, 0.18290441) }; 
\addplot [preaction={fill, Jaccard-color}, pattern=checkerboard] coordinates {(1, 0.32407927) }; 
\addplot [preaction={fill, AnchorGNN-color}, pattern=horizontal lines] coordinates {(1, 0.16141442) }; 
\addplot [preaction={fill, GEBEp-color}, pattern=vertical lines] coordinates {(1, 0.49537923) }; 
\addplot [preaction={fill, BiANE-color}, pattern=bricks] coordinates {(1, 0.53803584) }; 
\addplot [preaction={fill, RW_Uniform-color}, pattern=north east lines, postaction={pattern=dots}] coordinates {(1, 0.36945719) }; 
\addplot [preaction={fill, IDBR-color}, pattern=grid, postaction={pattern=dots}] coordinates {(1, 0.15834228) }; 
\addplot [preaction={fill, BiNE-color}, pattern=horizontal lines, postaction={pattern=vertical lines}] coordinates {(1, 0.16027038) }; 
\end{axis}
\end{tikzpicture}\hspace{0mm}
}%
\subfloat[{\em CiteSeer}]{
\begin{tikzpicture}[scale=1]
\begin{axis}[
        height=\columnwidth/2.3,
        width=\columnwidth/1.6,
    xtick=\empty,
    ybar=1.6pt,
    bar width=0.24cm,
    enlarge x limits=true,
    ylabel={\em F1-score},
    xticklabel=\empty,
    yticklabel style = {font=\scriptsize},
    ymin=0,
    ymax=0.45,
    every axis y label/.style={at={(current axis.north west)},right=7mm,above=0mm},
    legend style={at={(0.02,0.98)},anchor=north west,cells={anchor=west},font=\tiny}
    ]

\addplot [preaction={fill, AHPP-color}, pattern=north east lines] coordinates {(1, 0.40715178) }; 
\addplot [preaction={fill, BHPP-color}, pattern=dots] coordinates {(1, 0.33593756) }; 
\addplot [preaction={fill, HPP-color}, pattern=crosshatch dots] coordinates {(1, 0.33593756) }; 
\addplot [preaction={fill, PPR-color}, pattern=grid] coordinates {(1, 0.31330283) }; 
\addplot [preaction={fill, Simrank-color}, pattern=crosshatch] coordinates {(1, 0.33593756) }; 
\addplot [preaction={fill, Pearson-color}, pattern=north west lines] coordinates {(1, 0.22094560) }; 
\addplot [preaction={fill, Jaccard-color}, pattern=checkerboard] coordinates {(1, 0.24461118) }; 
\addplot [preaction={fill, AnchorGNN-color}, pattern=horizontal lines] coordinates {(1, 0.21871476) }; 
\addplot [preaction={fill, GEBEp-color}, pattern=vertical lines] coordinates {(1, 0.24327882) }; 
\addplot [preaction={fill, BiANE-color}, pattern=bricks] coordinates {(1, 0.22199952) }; 
\addplot [preaction={fill, RW_Uniform-color}, pattern=north east lines, postaction={pattern=dots}] coordinates {(1, 0.21132843) }; 
\addplot [preaction={fill, IDBR-color}, pattern=grid, postaction={pattern=dots}] coordinates {(1, 0.29906599) }; 
\addplot [preaction={fill, BiNE-color}, pattern=horizontal lines, postaction={pattern=vertical lines}] coordinates {(1, 0.21593501) }; 
\end{axis}
\end{tikzpicture}\hspace{0mm}
}%
\subfloat[{\em Blogcatalog}]{
\begin{tikzpicture}[scale=1]
\begin{axis}[
        height=\columnwidth/2.3,
        width=\columnwidth/1.6,
    xtick=\empty,
    ybar=1.6pt,
    bar width=0.24cm,
    enlarge x limits=true,
    ylabel={\em F1-score},
    xticklabel=\empty,
    yticklabel style = {font=\scriptsize},
    ymin=0,
    ymax=0.55,
    every axis y label/.style={at={(current axis.north west)},right=7mm,above=0mm},
    legend style={at={(0.02,0.98)},anchor=north west,cells={anchor=west},font=\tiny}
    ]

\addplot [preaction={fill, AHPP-color}, pattern=north east lines] coordinates {(1, 0.49547231) }; 
\addplot [preaction={fill, BHPP-color}, pattern=dots] coordinates {(1, 0.40355049) }; 
\addplot [preaction={fill, HPP-color}, pattern=crosshatch dots] coordinates {(1, 0.38239414) }; 
\addplot [preaction={fill, PPR-color}, pattern=grid] coordinates {(1, 0.25734528) }; 
\addplot [preaction={fill, Simrank-color}, pattern=crosshatch] coordinates {(1, 0.42374593) }; 
\addplot [preaction={fill, Pearson-color}, pattern=north west lines] coordinates {(1, 0.35711726) }; 
\addplot [preaction={fill, Jaccard-color}, pattern=checkerboard] coordinates {(1, 0.35399023) }; 
\addplot [preaction={fill, AnchorGNN-color}, pattern=horizontal lines] coordinates {(1, 0.17835505) }; 
\addplot [preaction={fill, GEBEp-color}, pattern=vertical lines] coordinates {(1, 0.39257329) }; 
\addplot [preaction={fill, BiANE-color}, pattern=bricks] coordinates {(1, 0.41548860) }; 
\addplot [preaction={fill, RW_Uniform-color}, pattern=north east lines, postaction={pattern=dots}] coordinates {(1, 0.33513029) }; 
\addplot [preaction={fill, IDBR-color}, pattern=grid, postaction={pattern=dots}] coordinates {(1, 0.17757329) }; 
\addplot [preaction={fill, BiNE-color}, pattern=horizontal lines, postaction={pattern=vertical lines}] coordinates {(1, 0.24185668) }; 
\end{axis}
\end{tikzpicture}\hspace{0mm}
}%
\subfloat[{\em Flickr}]{
\begin{tikzpicture}[scale=1]
\begin{axis}[
        height=\columnwidth/2.3,
        width=\columnwidth/1.6,
    xtick=\empty,
    ybar=1.6pt,
    bar width=0.24cm,
    enlarge x limits=true,
    ylabel={\em F1-score},
    xticklabel=\empty,
    yticklabel style = {font=\scriptsize},
    ymin=0,
    ymax=0.53,
    every axis y label/.style={at={(current axis.north west)},right=7mm,above=0mm},
    legend style={at={(0.02,0.98)},anchor=north west,cells={anchor=west},font=\tiny}
    ]

\addplot [preaction={fill, AHPP-color}, pattern=north east lines] coordinates {(1, 0.47489209) }; 
\addplot [preaction={fill, BHPP-color}, pattern=dots] coordinates {(1, 0.38299760) }; 
\addplot [preaction={fill, HPP-color}, pattern=crosshatch dots] coordinates {(1, 0.36558753) }; 
\addplot [preaction={fill, PPR-color}, pattern=grid] coordinates {(1, 0.20470024) }; 
\addplot [preaction={fill, Simrank-color}, pattern=crosshatch] coordinates {(1, 0.27697842) }; 
\addplot [preaction={fill, Pearson-color}, pattern=north west lines] coordinates {(1, 0.38503597) }; 
\addplot [preaction={fill, Jaccard-color}, pattern=checkerboard] coordinates {(1, 0.37270983) }; 
\addplot [preaction={fill, AnchorGNN-color}, pattern=horizontal lines] coordinates {(1, 0.11148681) }; 
\addplot [preaction={fill, GEBEp-color}, pattern=vertical lines] coordinates {(1, 0.40705036) }; 
\addplot [preaction={fill, BiANE-color}, pattern=bricks] coordinates {(1, 0.28532374) }; 
\addplot [preaction={fill, RW_Uniform-color}, pattern=north east lines, postaction={pattern=dots}] coordinates {(1, 0.29767386) }; 
\addplot [preaction={fill, IDBR-color}, pattern=grid, postaction={pattern=dots}] coordinates {(1, 0.20261391) }; 
\addplot [preaction={fill, BiNE-color}, pattern=horizontal lines, postaction={pattern=vertical lines}] coordinates {(1, 0.14860911) }; 
\end{axis}
\end{tikzpicture}
}%

\subfloat[{\em MovieLens}]{
\hspace{-7mm}
\begin{tikzpicture}[scale=1]
\begin{axis}[
        height=\columnwidth/2.3,
        width=\columnwidth/1.6,
    xtick=\empty,
    ybar=1.6pt,
    bar width=0.24cm,
    enlarge x limits=true,
    ylabel={\em F1-score},
    xticklabel=\empty,
    yticklabel style = {font=\scriptsize},
    ymin=0,
    ymax=0.35,
    every axis y label/.style={at={(current axis.north west)},right=7mm,above=0mm},
    legend style={at={(0.02,0.98)},anchor=north west,cells={anchor=west},font=\tiny}
    ]

\addplot [preaction={fill, AHPP-color}, pattern=north east lines] coordinates {(1, 0.30200031) }; 
\addplot [preaction={fill, BHPP-color}, pattern=dots] coordinates {(1, 0.12660079) }; 
\addplot [preaction={fill, HPP-color}, pattern=crosshatch dots] coordinates {(1, 0.12471673) }; 
\addplot [preaction={fill, PPR-color}, pattern=grid] coordinates {(1, 0.06100132) }; 
\addplot [preaction={fill, Simrank-color}, pattern=crosshatch] coordinates {(1, 0.16125552) }; 
\addplot [preaction={fill, Pearson-color}, pattern=north west lines] coordinates {(1, 0.15179648) }; 
\addplot [preaction={fill, Jaccard-color}, pattern=checkerboard] coordinates {(1, 0.17675347) }; 
\addplot [preaction={fill, AnchorGNN-color}, pattern=horizontal lines] coordinates {(1, 0.10775944) }; 
\addplot [preaction={fill, GEBEp-color}, pattern=vertical lines] coordinates {(1, 0.18674029) }; 
\addplot [preaction={fill, BiANE-color}, pattern=bricks] coordinates {(1, 0.10674494) }; 
\addplot [preaction={fill, RW_Uniform-color}, pattern=north east lines, postaction={pattern=dots}] coordinates {(1, 0.09193676) }; 
\addplot [preaction={fill, IDBR-color}, pattern=grid, postaction={pattern=dots}] coordinates {(1, 0.12036426) }; 
\addplot [preaction={fill, BiNE-color}, pattern=horizontal lines, postaction={pattern=vertical lines}] coordinates {(1, 0.11894598) }; 
\end{axis}
\end{tikzpicture}\hspace{0mm}
}%
\subfloat[{\em LastFM}]{
\begin{tikzpicture}[scale=1]
\begin{axis}[
       height=\columnwidth/2.3,
        width=\columnwidth/1.6,
    xtick=\empty,
    ybar=1.6pt,
    bar width=0.24cm,
    enlarge x limits=true,
    ylabel={\em F1-score},
    xticklabel=\empty,
    yticklabel style = {font=\scriptsize},
    ymin=0,
    ymax=0.6,
    every axis y label/.style={at={(current axis.north west)},right=7mm,above=0mm},
    legend style={at={(0.02,0.98)},anchor=north west,cells={anchor=west},font=\tiny}
    ]

\addplot [preaction={fill, AHPP-color}, pattern=north east lines] coordinates {(1, 0.56045502) }; 
\addplot [preaction={fill, BHPP-color}, pattern=dots] coordinates {(1, 0.09116287) }; 
\addplot [preaction={fill, HPP-color}, pattern=crosshatch dots] coordinates {(1, 0.06947329) }; 
\addplot [preaction={fill, PPR-color}, pattern=grid] coordinates {(1, 0.03368601) }; 
\addplot [preaction={fill, Simrank-color}, pattern=crosshatch] coordinates {(1, 0.09975695) }; 
\addplot [preaction={fill, Pearson-color}, pattern=north west lines] coordinates {(1, 0.01880237) }; 
\addplot [preaction={fill, Jaccard-color}, pattern=checkerboard] coordinates {(1, 0.11605843) }; 
\addplot [preaction={fill, AnchorGNN-color}, pattern=horizontal lines] coordinates {(1, 0.05047801) }; 
\addplot [preaction={fill, GEBEp-color}, pattern=vertical lines] coordinates {(1, 0.11139388) }; 
\addplot [preaction={fill, BiANE-color}, pattern=bricks] coordinates {(1, 0.03027028) }; 
\addplot [preaction={fill, RW_Uniform-color}, pattern=north east lines, postaction={pattern=dots}] coordinates {(1, 0.43567728) }; 
\addplot [preaction={fill, IDBR-color}, pattern=grid, postaction={pattern=dots}] coordinates {(1, 0.2837134) }; 
\addplot [preaction={fill, BiNE-color}, pattern=horizontal lines, postaction={pattern=vertical lines}] coordinates {(1, 0.2641941) }; 
\end{axis}
\end{tikzpicture}\hspace{0mm}
}%
\subfloat[{\em Google}]{
\begin{tikzpicture}[scale=1]
\begin{axis}[
        height=\columnwidth/2.3,
        width=\columnwidth/1.6,
    xtick=\empty,
    ybar=1.6pt,
    bar width=0.24cm,
    enlarge x limits=true,
    ylabel={\em F1-score},
    xticklabel=\empty,
    yticklabel style = {font=\scriptsize},
    ymin=0,
    ymax=0.6,
    every axis y label/.style={at={(current axis.north west)},right=7mm,above=0mm},
    legend style={at={(0.02,0.98)},anchor=north west,cells={anchor=west},font=\tiny}
    ]

\addplot [preaction={fill, AHPP-color}, pattern=north east lines] coordinates {(1, 0.55644553) }; 
\addplot [preaction={fill, BHPP-color}, pattern=dots] coordinates {(1, 0.47349222) }; 
\addplot [preaction={fill, HPP-color}, pattern=crosshatch dots] coordinates {(1, 0.47522403) }; 
\addplot [preaction={fill, PPR-color}, pattern=grid] coordinates {(1, 0.37638002) }; 
\addplot [preaction={fill, Simrank-color}, pattern=crosshatch] coordinates {(1, 0.45180094) }; 
\addplot [preaction={fill, Pearson-color}, pattern=north west lines] coordinates {(1, 0.42414079) }; 
\addplot [preaction={fill, Jaccard-color}, pattern=checkerboard] coordinates {(1, 0.42413717) }; 
\addplot [preaction={fill, AnchorGNN-color}, pattern=horizontal lines] coordinates {(1, 0.42405103) }; 
\addplot [preaction={fill, GEBEp-color}, pattern=vertical lines] coordinates {(1, 0.42450633) }; 
\addplot [preaction={fill, BiANE-color}, pattern=bricks] coordinates {(1, 0.42415165) }; 
\addplot [preaction={fill, RW_Uniform-color}, pattern=north east lines, postaction={pattern=dots}] coordinates {(1, 0.51433695) }; 
\addplot [preaction={fill, IDBR-color}, pattern=grid, postaction={pattern=dots}] coordinates {(1, 0) }; 
\addplot [preaction={fill, BiNE-color}, pattern=horizontal lines, postaction={pattern=vertical lines}] coordinates {(1, 0) }; 
\end{axis}
\end{tikzpicture}\hspace{0mm}
}%
\subfloat[{\em Amazon}]{
\begin{tikzpicture}[scale=1]
\begin{axis}[
        height=\columnwidth/2.3,
        width=\columnwidth/1.6,
    xtick=\empty,
    ybar=1.6pt,
    bar width=0.24cm,
    enlarge x limits=true,
    ylabel={\em F1-score},
    xticklabel=\empty,
    yticklabel style = {font=\scriptsize},
    ymin=0,
    ymax=0.7,
    every axis y label/.style={at={(current axis.north west)},right=7mm,above=0mm},
    legend style={at={(0.02,0.98)},anchor=north west,cells={anchor=west},font=\tiny}
    ]

\addplot [preaction={fill, AHPP-color}, pattern=north east lines] coordinates {(1, 0.64311775) }; 
\addplot [preaction={fill, BHPP-color}, pattern=dots] coordinates {(1, 0.58170830) }; 
\addplot [preaction={fill, HPP-color}, pattern=crosshatch dots] coordinates {(1, 0.63455001) }; 
\addplot [preaction={fill, PPR-color}, pattern=grid] coordinates {(1, 0.33006091) }; 
\addplot [preaction={fill, Simrank-color}, pattern=crosshatch] coordinates {(1, 0.36338136) }; 
\addplot [preaction={fill, Pearson-color}, pattern=north west lines] coordinates {(1, 0.34277199) }; 
\addplot [preaction={fill, Jaccard-color}, pattern=checkerboard] coordinates {(1, 0.34477247) }; 
\addplot [preaction={fill, AnchorGNN-color}, pattern=horizontal lines] coordinates {(1, 0.33383074) }; 
\addplot [preaction={fill, GEBEp-color}, pattern=vertical lines] coordinates {(1, 0.31400276) }; 
\addplot [preaction={fill, BiANE-color}, pattern=bricks] coordinates {(1, 0) }; 
\addplot [preaction={fill, RW_Uniform-color}, pattern=north east lines, postaction={pattern=dots}] coordinates {(1, 0.49294368) }; 
\addplot [preaction={fill, IDBR-color}, pattern=grid, postaction={pattern=dots}] coordinates {(1, 0) }; 
\addplot [preaction={fill, BiNE-color}, pattern=horizontal lines, postaction={pattern=vertical lines}] coordinates {(1, 0) }; 
\end{axis}
\end{tikzpicture}
}%
\end{small}
\caption{\textcolor{blue}{Clustering consistency validation of different methods.}}
\label{fig:cluster_consistency_validation}
\vspace{-1mm}
\end{figure*}

\begin{figure*}[t!]
\centering
\includegraphics[width=1\textwidth]{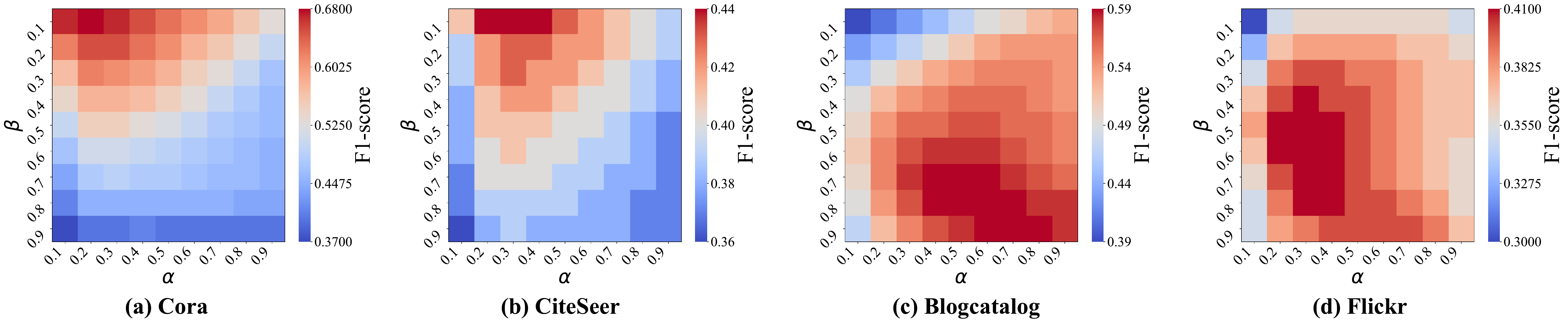}

\vspace{2mm}

\includegraphics[width=1\textwidth]{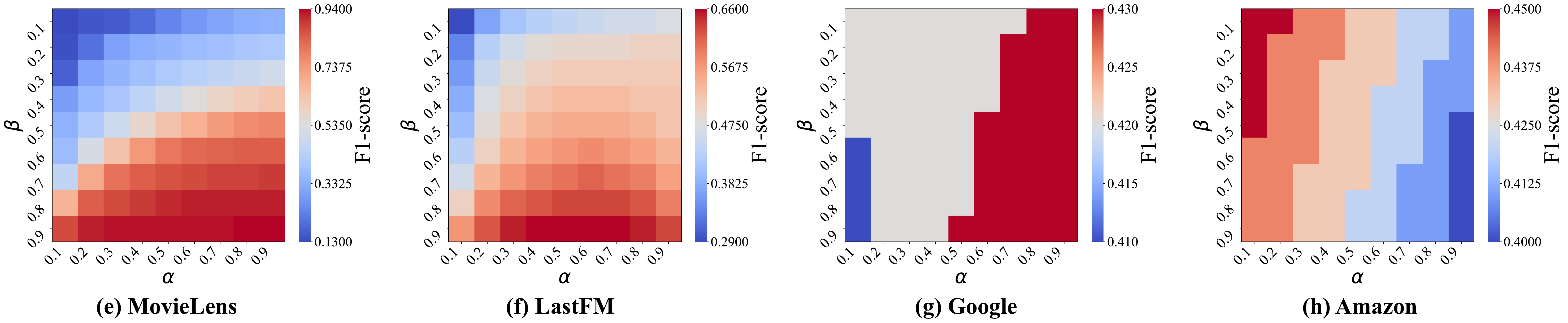}
\caption{\textcolor{blue}{Clustering consistency validation of our AHPP with varying  $\alpha$ and $\beta$.}}
\label{fig:alpha_beta_cluster}
\vspace{-2mm}
\end{figure*}

\stitle{Exp-4: Clustering consistency validation of different methods.} \textcolor{blue}{Given  a query node \( u\), cluster consistency validation aims to identify the nodes in the ground truth cluster associated with \( u \) based on their similarities. Namely, we first derive the ground truth clusters of nodes relevant to the query node for each dataset. Then, cluster consistency is evaluated using the widely recognized metric F1-score (\textit{F1-score@$k$}) \cite{liu2024bird}. The F1-score is the harmonic mean of precision and recall, providing a balanced measure of both. Precision is defined as the fraction of nodes in the top \( k \) positions that belong to the ground truth set, while recall is the ratio of the number of ground truth nodes in the top \( k \) positions to the total size of the ground truth set. For fairness, we set \( k \)  to the ground truth cluster size of the query node. Figure \ref{fig:cluster_consistency_validation} presents the F1-score results of AHPP and twelve other similarity measures across six datasets. The following observations can be made:  (1) The proposed AHPP consistently outperforms all other methods on all datasets. For example, on MovieLens and LastFM, AHPP achieves significant improvements of at least 11\% and 12\%, respectively, compared to state-of-the-art methods. (2) The bipartite network embedding-based methods (i.e., AnchorGNN, GEBEp, BiNE, IDBR and BiANE) perform poorly on most datasets. Additionally, BiNE, IDBR and BiANE even runs out of memory when applied to the Amazon dataset, which is the largest dataset among ours. In contrast, our AHPP exhibits excellent performance on every dataset and is capable of processing large-scale datasets. This effectively demonstrates the efficacy of our AHPP model in integrating attribute and structural information, the effectiveness of our algorithmic designs, and the scalability of our algorithms compared to bipartite network embedding methods. (3) BHPP and HPP demonstrate good performance on most datasets, with the exception of MovieLens and LastFM, which indicates their sensitivity to different datasets. The underlying reason is that in these two datasets, the ground-truth clusters are more closely related to node attributes than to structural information. (4) The unipartite graph similarity models (i.e., Pearson, Jaccard, SimRank, PPR, and RW\_Uniform) exhibit subpar performance in most cases. This is due to their inability to consider the unique structural properties of bipartite graphs. Moreover, the majority of them (i.e., Pearson, Jaccard, SimRank, and PPR) also neglect the attribute information of nodes. In summary, these findings offer compelling evidence that our AHPP consistently generates higher-quality clusters than the baselines.}

\stitle{Exp-5: Clustering consistency validation of our AHPP with varying $\alpha$ and $\beta$.} \textcolor{blue}{This experiment investigates the impact of the restart probability $\alpha$ and the attribute jumping probability $\beta$ of AHPP on the cluster consistency validation task, with a fixed value of $\epsilon=10^{-6}$. Figure \ref{fig:alpha_beta_cluster} illustrates the results, where the $x$-axis and $y$-axis represent the values of $\alpha$ and $\beta$, respectively. The color intensity in each cell of the heatmap corresponds to the F1-score for the respective parameter combination. Higher intensity indicates better F1-score values, allowing easy identification of parameter pairs that yield optimal or suboptimal performance. We make the following observations: (1) On most datasets, the best F1-score is achieved when $\alpha$ ranges from 0.1 to 0.3, except on Google, where it peaks with $\alpha$ values between 0.5 and 0.9. Since $\alpha$ represents the restart probability in AHPP, lower values encourage random walkers to explore neighboring nodes rather than remain at the current node. These results suggest that high-order information between nodes is particularly critical for the cluster consistency validation task on attributed bipartite graphs. (2) On MovieLens and LastFM, an increase in $\beta$ corresponds to higher F1-scores, suggesting that attribute information plays a crucial role in these datasets. This is consistent with the fact that both MovieLens and LastFM are user-item graphs enriched with informative node attributes. On other datasets, the best F1-score is achieved when $\beta$ ranges from 0.1 to 0.5, reflecting the necessity of balancing structural and attribute information. (3) The optimal values of $\alpha$ and $\beta$ vary across datasets, reflecting their distinct characteristics. This underscores the importance of striking a balance between high-order and local information, as well as between structural and attribute data, to maximize F1-scores.}

\stitle{Exp-6: Top-$k$ precision of our AHPP with varying $k$.} \textcolor{blue}{Following previous work on evaluating empirical precision \cite{wang2020personalized, wang2017fora}, we also use the well-known top-$k$ precision to further  evaluate the query quality of AHPP. It should be noted that other models (e.g., BHPP and HPP) are disregarded as their computational outputs are the BHPP values and HPP values, respectively, which deviate from ours (i.e., AHPP values). Specifically, given the ground truth top-k node set $U_k$ and the approximate set $\hat{U}_k$, the precision is defined as $\left| U_k \cap \hat{U}_k \right| / \left| U_k \right|$, representing the proportion of nodes in $\hat{U}_k$ that match those in the ground truth set $U_k$. For obtaining the ground truth set $U_k$, we run the PI method with a sufficiently large truncation step $T=150$ on each query and vary $k$ from 200 to 600. For all methods, we fix $\epsilon$ to $10^{-6}$. Figure \ref{fig:top_k_precision} presents the top-$k$ precision results for all methods across six datasets. We make the following observations: (1) Generally, ASRP, APP, FP, and PI achieve high-quality results consistently across all graphs and various $k$ settings. This demonstrates that our proposed APP and ASRP algorithms provide accurate estimations of the AHPP values. (2) ASRP produces the best precision results in almost all cases. Besides, as shown in Section \ref{subsec:efficiency}, ASRP outperforms all competitors in terms of running efficiency. This demonstrates that our ASRP algorithm strikes a good balance between accuracy and efficiency. In summary, both our APP and ASRP algorithms achieve strong results in approximate AHPP computation, confirming their correctness and query efficacy.}

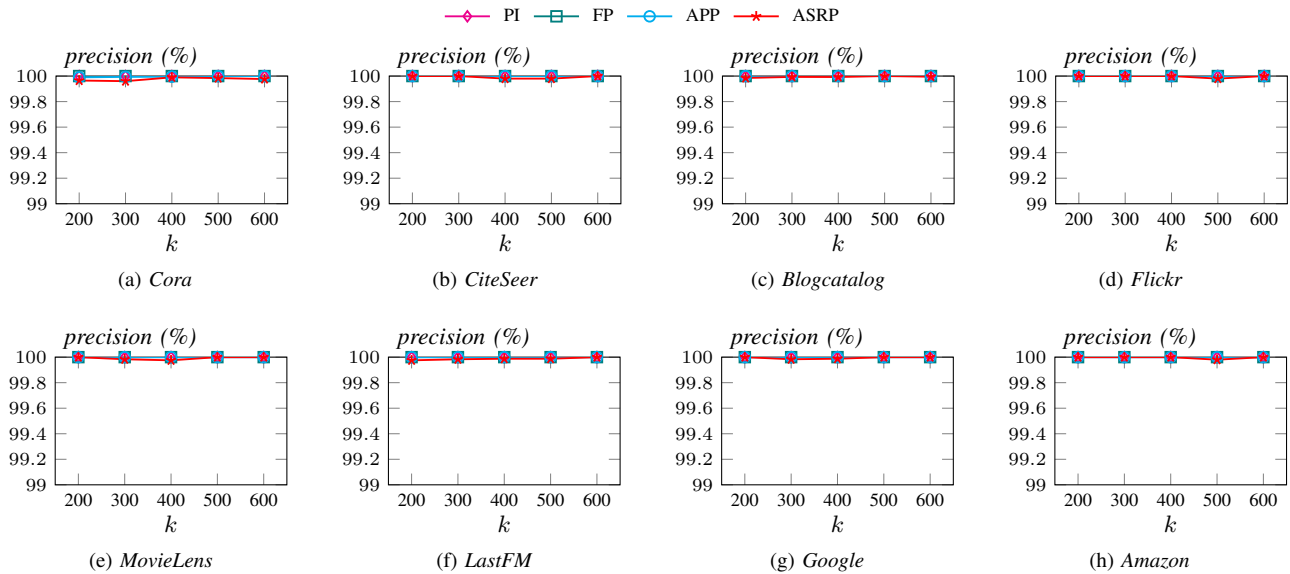
\begin{figure*}[t!]
\centering
\begin{small}
\begin{tikzpicture}
    \begin{customlegend}[legend columns=4,
        legend entries={PI, FP, APP, ASRP},
        legend style={at={(0.45,1.05)},anchor=north,draw=none,font=\scriptsize,column sep=0.1cm}]
    \addlegendimage{line width=0.25mm,mark=diamond,color=PI-color}
    \addlegendimage{line width=0.25mm,mark=square,color=FP-color}
    \addlegendimage{line width=0.25mm,mark=o,color=first-color}
    \addlegendimage{line width=0.25mm,mark=star,color=second-color}
    \end{customlegend}
\end{tikzpicture}
\\[-\lineskip]
\vspace{-3mm}
\subfloat[{\em Cora}]{
\hspace{-2mm}
\begin{tikzpicture}[scale=1]
    \begin{axis}[
        height=\columnwidth/2.7,
        width=\columnwidth/1.9,
        ylabel={\em precision (\%)},
        xlabel=$k$,
        xmin=0.5, xmax=5.5,
        ymin=99, ymax=100,
        xtick={1,2,3,4,5},
        xticklabel style = {font=\scriptsize},
        yticklabel style = {font=\scriptsize},
        xticklabels={200, 300, 400, 500, 600},
        scaled y ticks = false,
        every axis y label/.style={at={(current axis.north west)},right=10mm,above=0mm},
        every axis x label/.style={at={(current axis.south)}, below=3mm, anchor=north},
    ]
    \addplot[line width=0.25mm,mark=diamond,color=PI-color]  
        plot coordinates {
(1,	100)
(2,	100)
(3, 100)
(4,	100)
(5,	100)
    };
    
    \addplot[line width=0.25mm,mark=square,color=FP-color]  
        plot coordinates {
(1,	100)
(2,	100)
(3,	100)
(4,	100)
(5,	100)
    };

    \addplot[line width=0.25mm,mark=o,color=first-color]  
        plot coordinates {
(1,	99.990000)
(2,	99.993333)
(3,	99.995000)
(4,	100)
(5,	100)
    };

    \addplot[line width=0.25mm,mark=star,color=second-color]  
        plot coordinates {
(1,	99.965000)
(2,	99.960000)
(3,	99.990000)
(4,	99.984000)
(5,	99.976667)
    };
    \end{axis}
\end{tikzpicture}\hspace{4mm}%
}%
\subfloat[{\em CiteSeer}]{
\begin{tikzpicture}[scale=1]
    \begin{axis}[
        height=\columnwidth/2.7,
        width=\columnwidth/1.9,
        ylabel={\em precision (\%)},
        xlabel=$k$,
        xmin=0.5, xmax=5.5,
        ymin=99, ymax=100,
        xtick={1,2,3,4,5},
        xticklabel style = {font=\scriptsize},
        yticklabel style = {font=\scriptsize},
        xticklabels={200, 300, 400, 500, 600},
        scaled y ticks = false,
        every axis y label/.style={at={(current axis.north west)},right=10mm,above=0mm},
        every axis x label/.style={at={(current axis.south)}, below=3mm, anchor=north},
    ]
    \addplot[line width=0.25mm,mark=diamond,color=PI-color]  
        plot coordinates {
(1,	100)
(2,	100)
(3, 100)
(4,	100)
(5,	100)
    };
    
    \addplot[line width=0.25mm,mark=square,color=FP-color]  
        plot coordinates {
(1,	100)
(2,	100)
(3,	100)
(4,	100)
(5,	100)
    };

    \addplot[line width=0.25mm,mark=o,color=first-color]  
        plot coordinates {
(1,	100)
(2,	100)
(3,	100)
(4,	100)
(5,	100)
    };

    \addplot[line width=0.25mm,mark=star,color=second-color]  
        plot coordinates {
(1,	100)
(2,	100)
(3,	99.98)
(4,	99.98)
(5,	100)
    };
    \end{axis}
\end{tikzpicture}\hspace{4mm}%
}%
\subfloat[{\em Blogcatalog}]{
\begin{tikzpicture}[scale=1]
    \begin{axis}[
        height=\columnwidth/2.7,
        width=\columnwidth/1.9,
        ylabel={\em precision (\%)},
        xlabel=$k$,
        xmin=0.5, xmax=5.5,
        ymin=99, ymax=100,
        xtick={1,2,3,4,5},
        xticklabel style = {font=\scriptsize},
        yticklabel style = {font=\scriptsize},
        xticklabels={200, 300, 400, 500, 600},
        scaled y ticks = false,
        every axis y label/.style={at={(current axis.north west)},right=10mm,above=0mm},
        every axis x label/.style={at={(current axis.south)}, below=3mm, anchor=north},
    ]
    \addplot[line width=0.25mm,mark=diamond,color=PI-color]  
        plot coordinates {
(1,	100)
(2,	100)
(3, 100)
(4,	100)
(5,	100)
    };
    
    \addplot[line width=0.25mm,mark=square,color=FP-color]  
        plot coordinates {
(1,	100)
(2,	100)
(3,	100)
(4,	100)
(5,	100)
    };

    \addplot[line width=0.25mm,mark=o,color=first-color]  
        plot coordinates {
(1,	100)
(2,	100)
(3,	100)
(4,	100)
(5,	100)
    };

    \addplot[line width=0.25mm,mark=star,color=second-color]  
        plot coordinates {
(1,	99.985000)
(2,	99.993333)
(3,	99.992500)
(4,	100.000000)
(5,	99.995000)
    };
    \end{axis}
\end{tikzpicture}\hspace{4mm}%
}%
\subfloat[{\em Flickr}]{
\begin{tikzpicture}[scale=1]
    \begin{axis}[
        height=\columnwidth/2.7,
        width=\columnwidth/1.9,
        ylabel={\em precision (\%)},
        xlabel=$k$,
        xmin=0.5, xmax=5.5,
        ymin=99, ymax=100,
        xtick={1,2,3,4,5},
        xticklabel style = {font=\scriptsize},
        yticklabel style = {font=\scriptsize},
        xticklabels={200, 300, 400, 500, 600},
        scaled y ticks = false,
        every axis y label/.style={at={(current axis.north west)},right=10mm,above=0mm},
        every axis x label/.style={at={(current axis.south)}, below=3mm, anchor=north},
    ]
    \addplot[line width=0.25mm,mark=diamond,color=PI-color]  
        plot coordinates {
(1,	100)
(2,	100)
(3,	100)
(4,	100)
(5,	100)
    };
    
    \addplot[line width=0.25mm,mark=square,color=FP-color]  
        plot coordinates {
(1,	100)
(2,	100)
(3,	100)
(4,	100)
(5,	100)
    };

    \addplot[line width=0.25mm,mark=o,color=first-color]  
        plot coordinates {
(1,	100)
(2,	100)
(3,	100)
(4,	100)
(5,	100)
    };

    \addplot[line width=0.25mm,mark=star,color=second-color]  
        plot coordinates {
(1,	100)
(2,	100)
(3,	100)
(4,	99.98)
(5,	100)
    };
    \end{axis}
\end{tikzpicture}
}%

\subfloat[{\em MovieLens}]{
\begin{tikzpicture}[scale=1]
    \begin{axis}[
        height=\columnwidth/2.7,
        width=\columnwidth/1.9,
        ylabel={\em precision (\%)},
        xlabel=$k$,
        xmin=0.5, xmax=5.5,
        ymin=99, ymax=100,
        xtick={1,2,3,4,5},
        xticklabel style = {font=\scriptsize},
        yticklabel style = {font=\scriptsize},
        xticklabels={200, 300, 400, 500, 600},
        scaled y ticks = false,
        every axis y label/.style={at={(current axis.north west)},right=10mm,above=0mm},
        every axis x label/.style={at={(current axis.south)}, below=3mm, anchor=north},
    ]
    \addplot[line width=0.25mm,mark=diamond,color=PI-color]  
        plot coordinates {
(1,	100)
(2,	100)
(3,	100)
(4,	100)
(5,	100)
    };
    
    \addplot[line width=0.25mm,mark=square,color=FP-color]  
        plot coordinates {
(1,	100)
(2,	100)
(3,	100)
(4,	100)
(5,	100)
    };

    \addplot[line width=0.25mm,mark=o,color=first-color]  
        plot coordinates {
(1,	100)
(2,	100)
(3,	100)
(4,	100)
(5,	100)
    };

    \addplot[line width=0.25mm,mark=star,color=second-color]  
        plot coordinates {
(1,	100)
(2,	99.983333)
(3,	99.975)
(4, 100)
(5,	100)
    };
    \end{axis}
\end{tikzpicture}\hspace{4mm}%
}%
\subfloat[{\em LastFM}]{
\begin{tikzpicture}[scale=1]
    \begin{axis}[
        height=\columnwidth/2.7,
        width=\columnwidth/1.9,
        ylabel={\em precision (\%)},
        xlabel=$k$,
        xmin=0.5, xmax=5.5,
        ymin=99, ymax=100,
        xtick={1,2,3,4,5},
        xticklabel style = {font=\scriptsize},
        yticklabel style = {font=\scriptsize},
        xticklabels={200, 300, 400, 500, 600},
        scaled y ticks = false,
        every axis y label/.style={at={(current axis.north west)},right=10mm,above=0mm},
        every axis x label/.style={at={(current axis.south)}, below=3mm, anchor=north},
    ]
    \addplot[line width=0.25mm,mark=diamond,color=PI-color]  
        plot coordinates {
(1,	100)
(2,	100)
(3,	100)
(4,	100)
(5,	100)
    };
    
    \addplot[line width=0.25mm,mark=square,color=FP-color]  
        plot coordinates {
(1,	100)
(2,	100)
(3,	100)
(4,	100)
(5,	100)
    };

    \addplot[line width=0.25mm,mark=o,color=first-color]  
        plot coordinates {
(1,	100)
(2,	100)
(3,	100)
(4,	100)
(5,	100)
    };

    \addplot[line width=0.25mm,mark=star,color=second-color]  
        plot coordinates {
(1,	99.975)
(2,	99.983333)
(3,	99.9875)
(4,	99.9875)
(5,	100)
    };
    \end{axis}
\end{tikzpicture}\hspace{4mm}%
}%
\subfloat[{\em Google}]{
\begin{tikzpicture}[scale=1]
    \begin{axis}[
        height=\columnwidth/2.7,
        width=\columnwidth/1.9,
        ylabel={\em precision (\%)},
        xlabel=$k$,
        xmin=0.5, xmax=5.5,
        ymin=99, ymax=100,
        xtick={1,2,3,4,5},
        xticklabel style = {font=\scriptsize},
        yticklabel style = {font=\scriptsize},
        xticklabels={200, 300, 400, 500, 600},
        scaled y ticks = false,
        every axis y label/.style={at={(current axis.north west)},right=10mm,above=0mm},
        every axis x label/.style={at={(current axis.south)}, below=3mm, anchor=north},
    ]
    \addplot[line width=0.25mm,mark=diamond,color=PI-color]  
        plot coordinates {
(1,	100)
(2,	100)
(3,	100)
(4,	100)
(5,	100)
    };
    
    \addplot[line width=0.25mm,mark=square,color=FP-color]  
        plot coordinates {
(1,	100)
(2,	100)
(3,	100)
(4,	100)
(5,	100)
    };

    \addplot[line width=0.25mm,mark=o,color=first-color]  
        plot coordinates {
(1,	100)
(2,	100)
(3,	100)
(4,	100)
(5,	100)
    };

    \addplot[line width=0.25mm,mark=star,color=second-color]  
        plot coordinates {
(1,	100)
(2,	99.98333)
(3,	99.9875)
(4,	100)
(5,	100)
    };
    \end{axis}
\end{tikzpicture}\hspace{4mm}%
}%
\subfloat[{\em Amazon}]{
\begin{tikzpicture}[scale=1]
    \begin{axis}[
        height=\columnwidth/2.7,
        width=\columnwidth/1.9,
        ylabel={\em precision (\%)},
        xlabel=$k$,
        xmin=0.5, xmax=5.5,
        ymin=99, ymax=100,
        xtick={1,2,3,4,5},
        xticklabel style = {font=\scriptsize},
        yticklabel style = {font=\scriptsize},
        xticklabels={200, 300, 400, 500, 600},
        scaled y ticks = false,
        every axis y label/.style={at={(current axis.north west)},right=10mm,above=0mm},
        every axis x label/.style={at={(current axis.south)}, below=3mm, anchor=north},
    ]
    \addplot[line width=0.25mm,mark=diamond,color=PI-color]  
        plot coordinates {
(1,	100)
(2,	100)
(3,	100)
(4,	100)
(5,	100)
    };
    
    \addplot[line width=0.25mm,mark=square,color=FP-color]  
        plot coordinates {
(1,	100)
(2,	100)
(3,	100)
(4,	100)
(5,	100)
    };

    \addplot[line width=0.25mm,mark=o,color=first-color]  
        plot coordinates {
(1,	100)
(2,	100)
(3,	100)
(4,	100)
(5,	100)
    };

    \addplot[line width=0.25mm,mark=star,color=second-color]  
        plot coordinates {
(1,	100)
(2,	100)
(3,	100)
(4,	99.98)
(5,	100)
    };
    \end{axis}
\end{tikzpicture}
}
\vspace{0mm}
\end{small}
\caption{\textcolor{blue}{Top-$k$ precision of our AHPP with varying $k$ (MC is excluded due to timeouts on all datasets).}}
\label{fig:top_k_precision}
\vspace{-3mm}
\end{figure*}

\input{figures/experiment/Effectiveness/link_prediction}

\stitle{Exp-7: Link prediction of different methods.} \textcolor{blue}{Following existing studies \cite{huang2020biane, DBLP:conf/sigir/LiWLLJD24}, for each dataset, we randomly remove 20\% of the total edges. The missing links are predicted as follows: given an endpoint \( u \in U \) of a removed edge, we retrieve the top-\( k \) similar nodes in \( U \) of \( u \), and add edges between \( u \) and neighbors of its similar nodes to a set of candidates. Precision is used to measure the quality, defined as \( \frac{| \text{candidates} \cap \text{deleted edges} |}{| \text{deleted edges} |} \). We also vary \( k \) from 10 to 100 in steps of 10. The results are shown in Figure \ref{fig:link_prediction} (note that  Google and Amazon are excluded because of timeouts). We make the following observations: (1) The proposed AHPP consistently outperforms all other competitors across all datasets under different $k$. (2) The bipartite network embedding-based methods (i.e., AnchorGNN, GEBEp, BiNE, IDBR, and BiANE) generally exhibit inferior performance across most datasets. In particular, AnchorGNN and BiANE perform markedly worse than our approach on all evaluated datasets. These results clearly demonstrate the effectiveness and robustness of the proposed AHPP model for link prediction when compared with bipartite network embedding methods. (3) BHPP and HPP achieve strong performance on most datasets, except on Cora and CiteSeer. This observation suggests that these methods are sensitive to dataset characteristics; in particular, for citation networks such as Cora and CiteSeer, link formation is more strongly correlated with node attributes than with pure structural proximity. (4) The unipartite graph similarity models (i.e., Pearson, Jaccard, SimRank, PPR, and RW\_Uniform) show poor performance across most datasets. This stems from their inability to account for the unique structural properties of bipartite graphs. Moreover, the majority of them (i.e., Pearson, Jaccard, SimRank, and PPR) also overlook node attribute information. (5) Precision increases with \( k \), as a larger \( k \) corresponds to a more extensive candidate set for link prediction. To summarize, the empirical results further provide substantial validation of our AHPP's effectiveness in link prediction tasks.}

\vspace{-0.3cm}
\section{Related Work}\label{sec:relate}

\stitle{Similarity search on unipartite graphs} has been widely studied in recent years \cite{liao2022efficient, liao2023efficient, wang2021exactsim}. \textcolor{blue}{Notable examples include the structure-based models (e.g., Jaccard's coefficient \cite{li2024improved} and Pearson's correlation coefficient \cite{huang2024zhenmei}) and random walk-based models (e.g., SimRank \cite{wang2021exactsim} and Personalized PageRank (PPR) \cite{liao2023efficient}). This work is also closely related to PPR computation, as the proposed AHPP is the PPR on the graph constructed from an attributed bipartite graph. In the literature, many studies have focused on single-source PPR queries \cite{liu2016powerwalk, yoon2018tpa, wang2017fora, liao2022efficient, liao2023efficient,DBLP:journals/pacmmod/ZhouLLLW25}.} Among these works, \cite{yoon2018tpa} necessitated performing computationally expensive matrix-vector operations, \cite{liu2016powerwalk} employed numerous random walks to approximate PPR values, and several recent studies \cite{wang2017fora, liao2022efficient, liao2023efficient} aimed to enhance the efficiency of PPR computations by integrating local-push techniques with random walks. In addition to single-source PPR queries, numerous studies also focused on other PPR queries, including single-target PPR queries \cite{wang2020personalized, liao2022efficient}, single-pair PPR queries \cite{lofgren2016personalized, wang2016hubppr}, and top-k PPR queries \cite{wang2017fora, wei2018topppr}. Some recent studies have also been concerned with computing the PPR values on dynamic graphs \cite{hou2023personalized, li2023everything} and in parallel or distributed settings \cite{lin2019distributed, hou2021massively}, which are orthogonal to ours. \textcolor{blue}{However, these approaches focus solely on general unipartite graphs, without considering the complex yet highly useful structures of bipartite graphs.}

\stitle{Similarity search on bipartite graphs} has received significant attention recently. In the literature, \cite{jeh2002simrank} proposed the well-known SimRank method for similarity search in bipartite graphs (note that in the original paper, SimRank could also solve the unipartite graph case). Although SimRank was highly effective, its computational costs posed challenges when scaling to massive graphs. \cite{antonellis2008simrank++} proposed an improved version of SimRank, called SimRank++, and demonstrated its effectiveness in query rewriting. \textcolor{blue}{A recent advancement, P-Simrank \cite{dey2020p}, built upon the idea of SimRank to address scale-free bipartite networks. Nonetheless, these SimRank-like measures were hindered by high computational costs, which negatively impacted result quality. Earlier studies \cite{epasto2014reduce} introduced a promising proximity measure, HPP, which generalizes the idea of PPR \cite{andersen2006local} to bipartite networks within a unified iterative framework. Additionally, \cite{yang2022efficient} proposed Bidirectional Hidden Personalized PageRank (BHPP), a strengthened variant of HPP, and demonstrated superior query quality. \cite{liu2024bird} further designed a novel algorithm to improve query performance for single-source BHPP queries. However, they overlook the informative node attributes that frequently appear in applications.}


\section{CONCLUSION}\label{sec:conclude}
In this paper, we propose a novel random walk model of Attribute-augmented Hidden Personalized PageRank (AHPP), simultaneously capturing higher-order structural proximity and attribute similarity. We then develop two highly efficient push-style algorithms, APP and ASRP, accompanied by their theoretical analysis, to effectively solve the \( \epsilon \)-approximate single-source AHPP query problem. The APP algorithm is motivated by two key insights derived from the existing Forward Push technique, while the ASRP algorithm builds upon the APP algorithm by incorporating a synchronous push strategy and a more refined termination threshold. Finally, we conduct extensive experiments on real-world and synthetic datasets, comparing our approach with thirteen baselines to test the efficiency and effectiveness of our proposed solutions.

\bibliographystyle{IEEEtran}
\bibliography{main}

\end{document}